\documentclass{article}

\usepackage{microtype}
\usepackage{graphicx}
\usepackage{subfigure}
\usepackage{booktabs} % for professional tables
\usepackage{enumitem}
\usepackage{hyperref}

\usepackage[accepted]{icml2025}

\usepackage{amsmath}
\usepackage{amssymb}
\usepackage{mathtools}
\usepackage{amsthm}
\usepackage{xspace}

\newcommand{\eqdef}{\stackrel{\text{\tiny\rm def}}{=}}
\newcommand{\E}{\mathbf{E}}

\newcommand{\eps}{\varepsilon}
\newcommand{\prob}[1]{\Pr \left( #1 \right)}
\newcommand{\probgiven}[2]{\Pr \left( #1\ \middle|\ #2 \right)}
\newcommand{\EE}[2][]{%
  \if\relax\detokenize{#1}\relax
    \E\left[#2\right]%
  \else
    \E_{#1}\left[#2\right]%
  \fi
}

\newcommand{\ourPivot}{\textsc{Sparse-Pivot}\xspace}
\newcommand{\refClustering}{\textsc{Reference Clustering}\xspace}
\newcommand{\recompute}{\textsc{Recompute}\xspace}
\newcommand{\costEstimate}{\textsc{Cost-estimate}\xspace}

\newcommand{\withinClusterEstimate}{\textsc{In-Cluster-Cost-estimate}\xspace}

\newcommand{\insertNode}{\textsc{Insert-node-sparse-pivot}}

\newcommand{\ceil}[1]{\lceil #1 \rceil}

\newcommand{\rb}[1]{\left( #1 \right)}

\newcommand{\CCref}{C_{ref}}
\newcommand{\cC}{\mathcal{C}}

\newlength{\ourAlgIndent}
\newcommand{\ourIndent}[1][1]{%
    \foreach \n in {1,...,#1}{%
        \hspace{\ourAlgIndent}%
    }%
}

\DeclareMathOperator{\poly}{poly}
\DeclareMathOperator{\cost}{cost}
\newcommand{\cCnodel}{\cC_{\textsc{no-del}}}
\newcommand{\cCdel}{\cC_{\textsc{del}}}
\newcommand{\tcost}{\widetilde{cost}}

\usepackage[mathscr]{euscript}

\usepackage[capitalize,noabbrev]{cleveref}

\theoremstyle{plain}
\newtheorem{theorem}{Theorem}[section]

\newtheorem{lemma}[theorem]{Lemma}

\theoremstyle{definition}
\newtheorem{definition}[theorem]{Definition}

\theoremstyle{remark}
\newtheorem{remark}[theorem]{Remark}

\usepackage[disable,textsize=tiny]{todonotes}
\icmltitlerunning{Sparse-pivot: Dynamic correlation clustering for node insertions}

\begin{document}

\twocolumn[
\icmltitle{\textsc{Sparse-pivot:} Dynamic correlation clustering for node insertions}



\icmlsetsymbol{equal}{*}

\begin{icmlauthorlist}
\icmlauthor{Mina Dalirrooyfard}{equal,ms}
\icmlauthor{Konstantin Makarychev}{equal,nw}
\icmlauthor{Slobodan Mitrović}{equal,ucd}
\end{icmlauthorlist}

\icmlaffiliation{ucd}{University of California, Davis, USA. \texttt{smitrovic@ucdavis.edu}}

\icmlaffiliation{ms}{Machine Learning Research, Morgan Stanley, Canada. \texttt{minad@mit.edu}.}
\icmlaffiliation{nw}{Northwestern University, Chicago, USA. \texttt{konstantin@northwestern.edu}.}

\icmlcorrespondingauthor{Mina Dalirrooyfard}{minad@mit.edu}

\icmlkeywords{Machine Learning, ICML}

\vskip 0.3in
]



\printAffiliationsAndNotice{\icmlEqualContribution} % otherwise use the standard text.

\begin{abstract}
% Given an unweighted graph, the task of correlation clustering is to partition the node set into clusters to minimize the sum of the number of between-cluster edges and in-cluster non-edge node pairs.
% In this work, we study a dynamic version of this problem that handles node insertions and deletions.

We present a new Correlation Clustering algorithm for a dynamic setting where nodes are added one at a time. In this model, proposed by Cohen-Addad, Lattanzi, Maggiori, and Parotsidis (ICML 2024), the algorithm uses database queries to access the input graph and updates the clustering as each new node is added.
Our algorithm has the amortized update time of $O_{\eps}(\log^{O(1)}(n))$. Its approximation factor is $20+\varepsilon$, which is a substantial improvement over the approximation factor of the algorithm by Cohen-Addad et al. 
We complement our theoretical findings by empirically evaluating the approximation guarantee of our algorithm. The results show that it outperforms the algorithm by Cohen-Addad et al.~in practice.
% [Cohen-Addad, Lattanzi, Maggiori, Parotsidis, ICML'24]

\iffalse
the 5 approx (ref)algo runtime is linear in the size of degree, which could be as high as $O(n)$. We have an algorithm that has $O(log n)$ update time, and we show that it is a $4$-approx of the ref algorithm. The other algo that is sublinear has a very high approx. We evaluate the appproximation guarantee of our algorithm empiprically as well, using graphs with different densities and models, and show that its approximation factor outperforms both algorithms. 
\fi
\end{abstract}

\section{Introduction}
In this paper, we present a new dynamic algorithm with node updates for the Correlation Clustering problem with complete information.\footnote{By complete information, it is meant that for each pair of nodes there is information on their similarity.} 
Correlation Clustering is a well-studied problem that seeks to partition a set of objects into clusters based on their similarity. The problem is defined on a set of items represented as nodes in a graph, with similarity information provided through a set of edges. We assume that all pairs of nodes are classified as either ``similar" or ``dissimilar" by a noisy classifier. For every pair of similar nodes $u$ and $v$, there is an edge $(u, v)$ between them (sometimes referred to as a positive edge). For every pair of dissimilar nodes $u$ and $v$, no edge exists between them (such pairs are sometimes called negative edges). The objective is to find a clustering that minimizes the number of disagreements with the given edge set. An edge $(u, v)$ disagrees with the clustering if $u$ and $v$ are placed in different clusters, while a non-edge $(u, v)$ disagrees if $u$ and $v$ are placed in the same cluster.

The problem was introduced by \citet*{bansal2004correlation}. It can be easily solved if a disagreement-free clustering exists, as each cluster then corresponds to a connected component of the similarity graph $G$. However, when the classifier makes mistakes, and a disagreement-free solution does not exist, the problem becomes NP-hard. In their original paper, Bansal, Blum, and Chawla proposed a constant-factor approximation, which was subsequently improved in a series of works 
\citet{charikar2005clustering,demaine2006correlation,chawla2015near,cohen2022correlation,cohen2023handling,cao2024understanding}. In 2005, \citeauthor*{pivot} introduced combinatorial and LP-based algorithms with approximation factors of $3$ and $2.5$, respectively. The best-known approximation factor for Correlation Clustering with Complete Information is currently $1.437$ \citep*{cao2024understanding}.
Nevertheless, the combinatorial algorithm by \citet{pivot}, known as \textsc{Pivot}, remains one of the preferred choices in practice due to its simplicity and good empirical performance. 

Researchers have proposed various variants of \textsc{Pivot} that operate in parallel and streaming settings \cite{bonchi2014correlation,chierichetti2014correlation,pan2015parallel,cohen2021correlation,cambus2022parallel,BCMT,behnezhad2023single,cambus2022parallel,chakrabarty2023single}. Recently, there has been growing interest in algorithms that support dynamic updates. Consider a scenario where similarity information is received over time, requiring the clustering to be updated dynamically. \citet*{dalirrooyfardpruned} showed how to maintain a $(3+\varepsilon)$-approximation clustering with constant update time per edge insertion or deletion (see also papers by \citet{behnezhad2019fully} and \citet{chechik2019fully}).

\citet*{cohendynamic} proposed an algorithm designed for a setting where nodes and edges are stored in a database. Over time, new nodes are added to the database along with their incident edges. After nodes are added, the dynamic algorithm updates the existing clustering.
In this model, each edge is included in the database immediately after both its endpoints are inserted; no additional edges can be inserted or deleted afterward. The sequence of inserted nodes is non-adaptive, meaning it does not depend on the decisions of the clustering algorithm. Formally, we assume that the graph and the order of node arrivals are fixed in advance.
\citet{cohendynamic} showed how to achieve a constant (albeit very large) approximation with an update time of $\log^{O(1)} n$, measured in terms of database operations defined as follows: (1) retrieving the degree of a node $v$; (2) selecting a random neighbor of $v$; and (3) checking whether two nodes $u$ and $v$ are connected by an edge. The database model was introduced by \citet*{assadi2022sublinear}.
%Before their work, the only dynamic algorithm for correlation clustering with respect to node updates was that of \citet*{behnezhad2023single}, where they propose a simple pivot-based algorithm with approximation factor $(5+\eps)$ and running time 
The algorithm by \citet{cohendynamic} was the first dynamic algorithm for node insertions with sublinear update time. Although it provides a constant approximation, the proof suggests that the constant is very large (the paper does not estimate it). In this work, we propose a $(20+\varepsilon)$-approximation algorithm for Correlation Clustering with an update time of $O_{\varepsilon}(\log^{O(1)} n)$ database operations per node insertion.

In addition to node insertion, the algorithm by \citet{cohendynamic} supports deleting random nodes. Our algorithm supports a slightly weaker type of deletion: \emph{soft deletions} of random vertices. When a node is soft-deleted from the graph, it initially remains in the database but is marked as soft-deleted. Moreover, the classifier continues creating edges between newly arriving and soft-deleted nodes. Only when the algorithm requests their deletion are they purged from the database.

To motivate the model, consider the following example: An online store adds new items to its stock daily and aims to cluster all items based on similarity. Whenever a new item is added, the store runs a classifier to identify items similar to the new one. A record for the new item is then created and inserted into the database, along with edges connecting it to similar items. Instead of reclustering the entire dataset, the dynamic algorithm efficiently updates the clustering,  requiring only $O(\log^{O(1)} n)$ operations per item insertion.

%\mina{maybe a few sentences to explain before the example?}

\iffalse
\paragraph{Algorithm overview.} Our algorithm is based on the $5$-approximation variant of \textsc{Pivot}, developed by~\citet{behnezhad2023single} for the semi-streaming model. Their algorithm works as follows: First, it selects a random ordering $\pi$ of all nodes. For each node $u$, it picks the highest-ranked neighbor of $u$, denoted by $p(u)$, i.e., $p(u) = \arg\min_{w \in N(u)} \pi(w)$. The algorithm identifies all nodes $u$ such that $p(u) = u$, which we call pivots. Then, for every pivot, the algorithm creates a new cluster and assigns to it nodes $v$ with $p(v) = u$. All remaining nodes are placed in singleton clusters.

Our algorithm 
\fi

\section{Algorithm}

\iffalse
First we briefly discuss the algorithm of \cite{behnezhad2023single}\footnote{This algorithm is the dynamic version of $1$-pivot in ??}. The idea is that given a ordering $\pi$ on all the nodes, we define $p(u)= \min_{w\in N[u]} \pi(w)$. We call $u$ a pivot if $p(u)=u$, and each pivot has a cluster assigned to it. If for a node $v\neq u$, we have $p(v)=u$ where $v$ is a pivot, then we put $u$ in $v$'s cluster, and if $v$ is not a pivot we make $u$ a singleton. The dynamic version of this algorithm for node insertions is very similar and is depicted in 
\fi

Our algorithm is based on the $5$-approximation variant of \textsc{Pivot}, developed by~\citet*{behnezhad2023single} for the semi-streaming model. Their algorithm in the static setting works as follows: First, it selects a random ordering $\pi$ of all nodes. For each node $u$, it picks the neighbor of $u$ with the smallest rank and stores this neighbor in the pivot array $p$: $p(u) = \arg\min_{w \in N(u)} \pi(w)$. Next, the algorithm identifies all nodes $u$ such that $p(u) = u$, which we refer to as pivots.  We call $p(v)$ the pivot for node $v$. Then, the algorithm creates a new cluster for every pivot $u$ and assigns to it nodes $v$ with $p(v) = u$. All remaining unassigned nodes are placed in singleton clusters. Note that for each node $u$ in a singleton cluster, $p(u)$ remains equal to the neighbor of $u$ with the smallest rank. 
To reiterate, in this variant of \textsc{Pivot}, each vertex $u$ belongs to the cluster of $p(u)$ if $p(u) = p(p(u))$, and to a singleton cluster if $p(u) \neq p(p(u))$.

\citet*{behnezhad2023single} provide a dynamic edge-insertion version of this algorithm with constant update time.  
A dynamic implementation of this algorithm for node insertions can be easily inferred and is presented in \cref{alg:refclustering}. %Note that the update time for each node $u$ is linear in degree of $u$ and could be as high as $O(n)$. 
In this version, when the node $u$ being inserted is a pivot, the algorithm runs the \textsc{Explore} process. This process updates the pivot of each neighbor $w$ of $u$ and reassigns $w$ and its neighbors to the cluster of $u$ or singleton clusters, if necessary.
We refer to the clustering produced by this algorithm as \refClustering, as our algorithm aims to approximate this clustering.

\begin{algorithm}[h]
\caption{Insertion for \textsc{\refClustering} \label{alg:refclustering}}
\begin{algorithmic}[1]
    \STATE {\bfseries input} node $u$, graph $G$, ordering $\pi$, pivot array $p$.
    \STATE Compute $p(u)=\arg \min_{w\in N[u]} \pi(w)$.
    \STATE \textbf{if} $p(u)=u$:% and $\pi(u)\le \frac{L}{d(u)}$: \mina{the L constraint should be removed, but the pivot lemma should stay with L not L': because if we have this constraint, then what are we doing when $v=u$ and $\pi(u)\ge \frac{L}{d(u)}$? we are doing nothing in this case which doesn't make sense.}
    \STATE \ourIndent[1]  Mark $u$ as a pivot.
\STATE \ourIndent[1]     Create a new cluster with $u$ in it.
    \STATE \ourIndent[1] Run \textsc{explore}$(u,G,\pi,p)$
    \STATE \textbf{else if} $p(u)$ is a pivot:
    \STATE \ourIndent[1] Put $u$ in $p(u)$'s cluster.
    \STATE \textbf{else}:
    \STATE \ourIndent[1] Make $u$ a singleton.
    \end{algorithmic}
\end{algorithm}

\begin{algorithm}[tb]
\caption{\textsc{Explore} \label{alg:explore}}
\begin{algorithmic}[1]
    \STATE {\bfseries input} pivot $u$, graph $G$, ordering $\pi$, pivot array $p$.
    \STATE \textbf{for all} $w\in N[u]$:
    \STATE \ourIndent[1] \textbf{if} $\pi(p(w))>\pi(u)$:
    \STATE \ourIndent[2] \textbf{if} $w$ is a pivot, i.e., $p(w) = w$:
    \STATE \ourIndent[3] \textbf{for all} $z$ where $p(z)=w$: 
    \STATE \ourIndent[4] \textbf{if} $z \in N(u)$, then $p(z) \gets u$, 
    \STATE \ourIndent[4] \textbf{else} make $z$ a singleton
    % \STATE \ourIndent[4] $P\gets P\setminus \{w\}$.
    \STATE \ourIndent[2] $p(w)\gets u$.
    \end{algorithmic}
\end{algorithm}

%To make this algorithm dynamic, when a node $v$ arrives, it scans all its neighbors, and determines $p(v)$. If $p(v)=v$, then $v$ is a pivot and it creates its own cluster, and updates the pivot of all its neighbors. More formally, for each neighbor $w$, if $v$ has a lower rank compared to the current pivot of $w$, i.e. if $\pi(p(w))>\pi(v)$, then $p(w)$ is updated to be $v$, and $w$ goes to $v$'s cluster. 
%If $v$ is not a pivot, but $p(v)$ is a pivot, i.e. $p(p(v))=p(v)$, then we simply put $v$ in $p(v)$'s cluster. If $p(v)$ is not a pivot, we make $v$ a singleton.\mina{need names for our algorithm.} \mina{might be a good idea to put reference clustering in proper algo environment}

\subsection{Making the \refClustering algorithm faster}
%We call the algorithm of \cite{behnezhad2023single} reference-clustering since we are going to approximately make the same clusters but in a more efficient manner. 

The challenge with \refClustering is that, after each insertion, it scans the entire neighborhood of the inserted node $u$, which can be as large as $\Theta(n)$, where $n$ is the size of the current graph. This makes this \refClustering prohibitively expensive. To improve the efficiency, we perform this exhaustive search only for pivots. 
For non-pivot nodes $u$, we attempt to recover the unknown pivot $v$ by sampling $O(\log n)$ neighbors of $u$, examining the set of pivots of these sampled neighbors, and setting $p(u)$ to the neighbor with the smallest ranked among the neighbors scanned in this process (which could be the samples or their pivots).  
We show that this approach succeeds when the cluster $C$ in the reference clustering is sufficiently dense and $u$ does not have too many neighbors outside of $C$.

Instead of selecting a random ordering of vertices $\pi$, our algorithm assigns each value $\pi(u)$ uniformly at random from the interval $[0, 1]$. This simplifies maintaining a random ordering in dynamic settings, where the exact number of nodes is not known in advance and also makes the logic of our algorithm a bit simpler.

We present our algorithm for node insertion in \cref{alg:node-dynamic-pivot}.  
We use three main ideas:  

\noindent (1) For vertices $u$ with $\pi(u) \leq \frac{L}{d(u)}$, where $L = O(\log n)$, we run \textsc{explore} on $u$ if $u$ is a pivot, similarly to \refClustering. If $u$ is not a pivot, we run \textsc{explore} on its pivot to update its cluster. 

\noindent (2) For vertices $u$ with $\pi(u) > \frac{L}{d(u)}$, we find a pivot by examining the set of pivots of $\Theta(\log n)$ random neighbors of $u$ and selecting the neighbor of $u$ with smallest ranked in that set as a pivot for $u$.  

(3) In the obtained tentative clustering, we identify certain nodes, remove them from their clusters, and place them into singleton clusters. Specifically, we remove nodes whose in-cluster degree is below a threshold $t$. The optimal value of $t$ is determined by trying out $O(\log n)$ possible choices. In particular, when a new cluster is made upon arrival of a pivot node, the subroutine \textsc{break-cluster} removes some nodes from the cluster and makes them singletons. When the node inserted is not a pivot and is assigned to pivot $v$, the subroutine \textsc{update-cluster} updates the set of nodes assigned to $v$ that are put in singletons upon insertion of the new node (See \cref{sec:cost-estimate} for details).

We now provide some motivation for the algorithm. First, we observe that the expected time required for finding pivots in items (1) and (2) is $\log^{O(1)} n$. In item (1), we spend $d(u) \log n$ time with probability $\log n / d(u)$. In item (2), we always spend $O(\log n)$ time by examining $O(\log n)$ random neighbors and their pivots. 

We run the step described in item (3) after (roughly) every $\epsilon s$ insertions, where $s$ is the size of the cluster. In this step, we test $O(\log n)$ different thresholds $t$ and, for each choice of $t$, estimate the cost of the split defined by $t$ by sampling random edges within the cluster (see \cref{sec:cost-estimate} for details). Hence, the total running time of one such step is $s\log^{O(1)} n$, and the amortized cost per insertion is $\log^{O(1)} n$.

We now discuss the approximation factor. Item (1) ensures that all pivots in the reference clustering are marked as pivots by our algorithm with high probability. Specifically, for every pivot $u$ in the reference clustering, we have $\pi(u) \leq L/d(u)$ with probability at least $1 - \text{poly}(1/n)$, ensuring that \textsc{explore} is called on every pivot $u$. This follows from the observation that if $\pi(u) > L/d(u)$, then $\pi(v) > L/d(u)$ for all neighbors $v \in N[u]$. The probability of this event is at most  
$
(1 - L/d(u))^{d(u)} \leq \text{poly}(1/n).
$  

We define \emph{good} nodes as follows: loosely speaking, a node is good if it is connected to most nodes within its cluster and to relatively few outside of it.
We show that \emph{good} nodes in our clustering are assigned the same pivots as in the reference clustering (when our algorithm and \refClustering use the same ordering $\pi$).  For now, let us assume that all nodes are good. If a good node $u$ arrives after its pivot $w$, then $u$ will join $w$'s cluster when $w$ scans all its neighbors, including $u$. If $u$ arrives after $w$ while more than $25\%$ nodes in $w$'s cluster have yet to arrive, then for at least one node $v$ arriving after $u$, we will call \textsc{explore}, which will assign $u$ the correct pivot $w$. Here, we rely on the assumption that all nodes are good.
Finally, if $u$ arrives among the last $25\%$ of nodes in $w$'s cluster, then by the time $u$ arrives, the algorithm will have assigned the correct pivot $w$ to most of its neighbors in the reference cluster. Consequently, a high fraction of $u$'s neighbors will have $w$ as their pivot, and some of these neighbors will be included in the random sample of neighbors (see item (2)). Hence, $u$ will also be assigned the correct pivot. Our full analysis can be found in \cref{sec:analysis}.

Let us now consider bad (not good) nodes. These nodes incur a very substantial cost in the reference clustering. If we could remove them from their reference clusters and place them into singleton clusters, the overall clustering cost would not increase significantly—in fact, it might even decrease. Unfortunately, our algorithm cannot identify these bad nodes; as a result, they may join clusters other than their own reference clusters. This can substantially increase the cost of those clusters. To address this issue, we partition each tentative cluster into two parts (see item (3)): the first part remains a cluster, while the second part is broken into singleton clusters.

\cref{thm:approximation} provides the approximation guarantees of our algorithm which we call $\ourPivot$. 
\begin{theorem}\label{thm:approximation}
For any $\epsilon<1/1000$, the expected cost of $\ourPivot$ with parameter $\epsilon$ is at most $4(1+O(\epsilon))$ times the expected cost of the \refClustering.
    % Consider a node $u$. Then
    % \[
    %     \prob{\alg(u) \neq p(u)\ |\ \text{$p(u)$ is defined, $C(u)$ is good, $u$ is good}} < \frac{1}{\poly(n)}.
    % \]
    % In this theorem, $p(u)$ refers to the reference clustering.
\end{theorem}

Our algorithm for soft deletions is very simple: We ignore them! In fact, we recompute the whole clustering again after $\Theta(\eps) N$ many updates, where $N$ is the number of nodes when we last recomputed the clustering. The recomputation is only necessary for deletions. Our running time guarantees are provided in \cref{thm:running-time}.

\begin{theorem}\label{thm:running-time}
    Let $T$ be the total number of updates. With high probability, \cref{alg:node-dynamic-pivot} runs in amortized $\poly(\log T, \frac{1}{\eps})$ time.
\end{theorem}
%; if no deletions exist, no recomputation is needed. 
%  deletions until we have at least $\epsilon n$ many deletions. Then we recompute the whole clustering, using ??   \mina{describe + put it in algo environment. The proofs for it can go in the appendix.}

%\stodo{Explain somewhere that when we say $n$, we are referring to the number of nodes after the very last recompute.
%}
%\stodo{Actually, we can refer to the current number of nodes as well since we will be recomputing everything after $\eps n / 2$ updates.}\mina{i agree, and this will be easier. also even in the only insersion case without recompute, i think all of our proofs carry if we consider n to be the current num of nodes.}

\begin{algorithm}[h]
\caption{\insertNode \label{alg:node-dynamic-pivot}}

\begin{algorithmic}[1]
    \STATE \textbf{input} node $u$ to be inserted, current graph $G$ of size $n$, ordering $\pi$, pivot vector $p$.
        \STATE \ourIndent Let $\pi(u)\in [0,1]$ be chosen uniformly at random. %Let $P$ be the list of pivots. \mina{no need to keep P}
        \STATE \ourIndent Let $p(\cdot)$ indicate the lowest rank neighbor of each node.
        \STATE \ourIndent Let $B_v$ indicate the set of nodes with pivot $v$, and let $C_v\subseteq B_v$ be the nodes of $B_v$ that are in $v$'s cluster. 
%        \STATE \ourIndent Let $t_v$ be the degree threshold for pivot $v$.
        \STATE \ourIndent \textbf{if} $\pi(u)\le \frac{L}{d(u)}$:
        \STATE \ourIndent[2] Find $v=\arg \min_{w\in N[u]} \pi(w)$.
        \STATE \ourIndent[2] \textbf{if} $v=u$:% and $\pi(u)\le \frac{L}{d(u)}$: \mina{the L constraint should be removed, but the pivot lemma should stay with L not L': because if we have this constraint, then what are we doing when $v=u$ and $\pi(u)\ge \frac{L}{d(u)}$? we are doing nothing in this case which doesn't make sense.}
        \STATE \ourIndent[3] Make $u$ a pivot: $p(u) \gets u$, and make a new cluster with $u$ in it: $B_u = \{u\}$, $t_v = 0$% and $P\gets P\cup \{w\}$.
        \STATE \ourIndent[3] \textsc{explore}$(u,G,\pi,p)$
        \STATE \ourIndent[3]$C_u\gets$ \textsc{break-cluster}$(B_u)$.
        
        \STATE \ourIndent[2] \textbf{else if} $v\neq u$ and $v$ is a pivot:
        \STATE \ourIndent[3] $p(u)\gets v$, $B_v\gets B_v\cup\{u\}$.
        \STATE \ourIndent[3] $C_v\gets $\textsc{update-cluster}($u,B_v$)%\textbf{if} $d(u)\le t_v$ \mina{replace}, Make $u$ a singleton. Otherwise $C_v\gets C_v\cup \{u\}$.
        \STATE \ourIndent[3] \textbf{if} $d(v)\le \frac{L}{\pi(u)}$: %\snote{I flipped $d(u) \le \rho d(v)$ to $d(u) \ge \rho d(v)$}
        \STATE \ourIndent[4] \textsc{explore}$(v,G,\pi,p)$
        \STATE \ourIndent[2] \textbf{else if} $v\neq u$ and $v$ is not a pivot:
        \STATE \ourIndent[3] make $u$ a singleton.
        \STATE \ourIndent \textbf{else}:
        \STATE \ourIndent[2] Let $S$ be a $O(\log n)$-sized sample of $N[u]$. \label{line:sample-log-neighbors}
        \STATE \ourIndent[2] Let $s^* = \arg\min_{s\in S, \{p(s), u\}\in E(G)} \pi(p(s))$. Let $v:=p(s^*)$
        %\stodo{Do you maybe want to say \textbf{if} $\pi(v)<\pi(u)$ below?}
        \STATE \ourIndent[2] \textbf{if} $\pi(v)<\pi(u)$:
        \STATE \ourIndent[3] $p(u)\gets v$, $B_v\gets B_v\cup\{u\}$. 
        \STATE \ourIndent[3] $C_v\gets $\textsc{update-cluster}($u,B_v$)%\textbf{if} $d(u)\le t_v$ \mina{replace}, Make $u$ a singleton. Otherwise  $C_{v}\gets C_{v}\cup \{u\}$.
        \STATE \ourIndent[2] \textbf{if} $u$ did not get clustered:
        \STATE \ourIndent[3]\label{line:singleton} make $u$ a singleton, $u$ does not have a pivot.
%
        % \STATE \ourIndent \textbf{if} for a pivot $v$, $|C_v|$ has changed by more than a factor of $(1+\epsilon)$ since run of \textsc{break-cluster}$(C_v, v,t)$:
        % \STATE \ourIndent[2] $C_v\gets$ \textsc{break-cluster}$(C_v, v,t)$.
        % \STATE \ourIndent \textbf{if} for a pivot $v$, $|B_v|$ has changed by more than a factor of $(1+\epsilon)$ since run of \textsc{break-cluster}$(B_v, v)$:
        % \STATE \ourIndent[2] $C_v, t_v \gets$\textsc{break-cluster}$(B_v, v)$.       
    \end{algorithmic}
\end{algorithm}

\subsection{Analysis Preliminaries}
We use subscript $ref$ to refer to \refClustering. We fix time, and we compare the cost of our algorithm to the cost of \refClustering at this time. We refer to the current graph as $G$.  %The definitions below are with respect to the reference clustering, but they also can be defined with respect to any clustering algorithm.  
For a node $v$, ordering $\pi$ and clustering algorithm $A$, let $C_{A}^\pi(v)$ be the cluster of $v$ in $A$ with respect to the ordering $\pi$. We drop the superscript $\pi$ and subscript $A$ when it is clear from the context. The clustering algorithms that we consider throughout our analysis all define pivots for all non-singleton clusters. 

%The cost of a node $u$ in any clustering is defined as the number of non-neighbors in its cluster and half of the number of neighbors outside of its cluster among the nodes that arrived \emph{before} $u$.
Given a ordering $\pi$ and clustering algorithm $A$, let $p_{A}^{\pi}(u)$ denote the pivot of $u$ chosen by algorithm $A$. For a set of nodes $S$, let $d_S(v)$ be the degree of $v$ in $S$, and let $d(v)$ be the degree of $v$ in the graph $G$. We will classify nodes depending on how their neighborhoods intersect the cluster to which they are assigned.
\begin{definition}\label{def:light-poor-heavy-lost}
     Let $\alpha<1$ and $\beta>1$. Let $A$ be a clustering algorithm. For a fixed ordering $\pi$, we call vertex $u$ 
     \begin{itemize}
         \item \textit{$A$-light} if $d_{C}(u) \le \frac{|C|}{3}$, where $C=C_{A}^\pi(u)$.
         \item  \textit{$(A,\alpha)$-poor} if $d(u)\le \alpha d(p_{A}(u))$ and $u$ is not light.%\mina{have to change it to $|N(u)\cap N(v)|\le \alpha d(p_{A}(u))$}
         \item \textit{$A$-heavy} if $d(u)\ge \beta |C|$, where $C=C_{A}^\pi(u)$.
         \item \textit{$A$-bad} if $u$ is $(A,3\alpha\beta)$-poor, $A$-heavy or $A$-light, and \textit{$A$-good} otherwise.
         \item \textit{$A$-lost} if the number of $A$-bad neighbors of $u$ in $C_{A}^\pi(u)$ is at least $\beta$ times the number of $A$-good neighbors of $u$ in $C_{A}^\pi(u)$.
     \end{itemize}
     We drop $A$- if the clustering algorithm is clear from the context.
     %if $d(u)\le \alpha d(p_{ref}(u))$, and it is \textit{heavy} if $d(u)\ge \beta |C_u|$ \mina{I think the definition of heavy node should be this: $d(u)\ge \beta d_{C_u}(u)$. Lemma \ref{lem:heavynode} will follow}. We call $u$ \textit{bad} if it is light or heavy, and we call it \textit{good} otherwise. \mina{we may want to have a separate name for nodes that have low degree inside the cluster wrt the size of the cluster (kinda like type 1 light}
\end{definition}

\begin{definition}[Poor clusters]
    For a fixed ordering $\pi$ and clustering algorithm $A$, let $C$ be a cluster with pivot $v$. 
    We call $C$ an $(A,\alpha)$-poor cluster if $C$ has at least one $(A,\alpha)$-poor node.
\end{definition}
\begin{definition}[Good and bad clusters]
    Let $\gamma<1$ be a constant. For a fixed ordering $\pi$ and clustering algorithm $A$, we call a cluster $C^\pi_A$ good if it is not $(A,\alpha)$-poor, and it has at least $\gamma|C^\pi_A|$ good nodes. 
    Otherwise, we call it bad.
\end{definition}
\begin{remark}
    Note that instead of fixing the clustering algorithm $A$ and ordering $\pi$, we can still have the above definitions if we fix the clusters and pivots. 
\end{remark}
\begin{definition}[Cost of a cluster]
    Given a clustering, the cost of a cluster $C$ is the number of non-edges inside $C$, and half of the number of edges with exactly one endpoint in $C$.
\end{definition}
The cost of a clustering is the sum of the cost of its clusters. 
% \begin{definition}[Bad cluster]
%     A cluster $C$ in the reference clustering is called \textit{bad} if at most $\gamma$ fraction of its vertices are good. 
% \end{definition}

\section{Analysis}\label{sec:analysis}
\subsection{Analysis outline}
First, we explain the intuition behind classifying nodes in \cref{def:light-poor-heavy-lost}. Let the cost of a node be half the number of non-neighbors it has in its cluster, plus half the number of neighbors it has outside the cluster. Note that the sum of the costs of nodes in a cluster equals the cost of the cluster. Consider $A$ to be the reference clustering in \cref{def:light-poor-heavy-lost}. A light node in cluster $C$ has a lot of non-edges attached to it in $C$, and a heavy node has a lot of edges attached to it that leave the cluster $C$. So both have a high cost. In fact, we show that if we make them singletons, the cost of the clustering does not change much. So, in a sense, in our algorithm, we do not care how $ref$-heavy or $ref$-light nodes are being clustered as long as their cost is somewhat comparable to their cost as singletons. 

Consider poor nodes. A poor node has a much lower degree than its pivot. We show that if a cluster has at least one poor node, then any node in this cluster that is not heavy or light must be poor (\cref{lem:all-poor}). Then we show that, in fact, if we make the whole cluster singleton, the cost of the clustering does not change much \emph{on average}. This is because the pivot of this cluster has a very high cost, and any node becomes the pivot of a poor cluster with low probability. 

In summary, we have shown that if we make the bad nodes (heavy, poor, or light) in the reference clustering singleton, the cost of the clustering does not change much (\cref{lem:bad-and-lost-single}). We further show that we can make all the nodes in a bad cluster singleton as well since the cost of this cluster is already too high. We call this clustering $ref'$.

We then show that in \ourPivot, not only do we detect pivots correctly with high probability (\cref{lem:pivot-set}), but also all the good nodes that are not lost, i.e., they do not have many bad neighbors, are assigned the correct pivot. We show this in two parts: Consider a pivot $v$, and suppose $C$ is the set of all the $ref$-good nodes with pivot $v$ that are not $ref$-lost. First, we show that if a good node $u\in C$ arrives rather early compared to other nodes in $C$, at some point its pivot runs \textsc{explore} and it detects if $u$ is clustered incorrectly (\cref{lem:1-eps-good-nodes}). If $u$ arrives rather late, then the sampling procedure will hit one of the neighbors of $u$ in $C$ that is correctly clustered and so correctly assigns $v$ as the pivot of $u$ (\cref{lem:misplaced}).

Now since we cluster the $ref$-good nodes correctly with high probability, if we could detect $ref$-bad nodes in \ourPivot, we could make them singletons, and thus get $ref'$. However, instead, if $B_v$ is the set of nodes that have $v$ as their pivot, we make a dense subset $C_v$ of $B_v$ one cluster, and make the rest of $B_v$ singleton. This step is crucial since there might be many (bad or lost) nodes incorrectly assigned to $v$ that increase the number of the non-edges in $B_v$ significantly. We show that the cost of making $B_v\setminus C_v$ singleton is at most $4$ times the cost of making the $ref$-bad nodes in $B_v$ singleton (\cref{lem:4-approx-break-cluster}).

Now we begin our formal analysis.
Let $\beta\ge \frac{4+\epsilon}{\epsilon}$, $\alpha <\min(\frac{\epsilon}{24\beta}, \frac{1}{39\beta})$, and $\gamma \le \frac{\epsilon}{2}$. Let $L\ge \frac{4c\beta}{\epsilon \gamma\alpha}\log n$ for some arbitrary large constant $c$. Let the number of samples when the inserted $u$ satisfies $\pi(u)>L/d(u)$ be at least $100 \log (\frac{1}{1-x})=O(\log n)$, where $x=(\frac{1}{\beta+1}-\epsilon)/\beta$.

\subsection{Making all the bad and lost nodes singleton in Reference clustering}

We show that if a cluster has one $\alpha$-poor node, then all the nodes in that cluster are $3\alpha\beta$-poor.
\begin{lemma}\label{lem:all-poor}%\mina{this should be referenced somewhere and its not}
    Let $\pi$ be an ordering and $A$ be a clustering algorithm. If $C_A^{\pi}$ is an $(A,\alpha)$-poor cluster, then any $u\in C_A^{\pi}$ which is not light or heavy is $(A,3\alpha\beta)$-poor.
\end{lemma}
\begin{proof}
    Let $v$ be the pivot of $C=C_A^{\pi}$. 
    First, since $C$ is an $(A,\alpha)$-poor cluster there is a node $w\in C$ that is $(A,\alpha)$-poor. 
    Since, by definition of poor nodes, $w$ is not light, we have that $|C|/3\le d(u)\le \alpha d(v)$. So $3\alpha d(v)\ge |C|$. Now for any $u$ that is not heavy we have $d(u)\le \beta |C|\le 3\alpha\beta d(v)$. Thus if $u$ is not light, then $u$ is $(A,3\alpha\beta)$-poor. 
\end{proof}

The following lemma shows that if, in \refClustering, we make all the bad nodes and lost nodes singleton, the cost only increases by a factor of $(1+O(\epsilon))$.

\begin{lemma}[Cost of making bad and lost nodes singletons]\label{lem:bad-and-lost-single}
    Consider a clustering algorithm $A$ where the probability of any node $v$ being a pivot is at most $1/d(v)$. 
    Let $B$ be the algorithm that first runs $A$, and then makes $A$-heavy, $A$-light, $(A,3\alpha\beta)$-poor and $A$-lost nodes singletons. 
    Then $\EE[\pi]{cost_B}\le (1+7\epsilon)\EE[\pi]{cost_A}$
    
    % \[
    %     \EE[\pi]{cost_B} \le \frac{\beta+1}{\beta-1} \cdot \rb{1+4\alpha\frac{1+3/2 \cdot \alpha}{1-5/2 \cdot \alpha}}\cdot \frac{\beta+3}{\beta-3} \EE[\pi]{cost_A}.
    % \]
    %expected cost of $B$ is at most $(\frac{\beta+1}{\beta-1})^2 ??$ the cost of $A$.
\end{lemma}
\begin{proof}
    Let $\alpha' = 3\alpha\beta$, and let $A'$ be the algorithm that runs $A$ and then makes the $A$-heavy and $A$-light nodes singleton. Let $A''$ be the algorithm that runs $A'$, and makes all the $(A,\alpha')$-poor nodes singleton. 
    For the clarity and brevity of notation, we use
        $\delta \eqdef 4\alpha'\frac{1+3/2 \cdot \alpha'}{1-5/2 \cdot \alpha'}$.
    For any fixed ordering $\pi$, by \cref{lem:heavynode}, $cost_{A'}^\pi\le \frac{\beta+1}{\beta-1}cost_A^{\pi}$, and
    \[
        cost_{A'} \le  \frac{\beta+1}{\beta-1}cost_A.
    \]
    
    Now note that the probability of any node $v$ being a pivot in $A'$ is at most the probability of any node being a pivot in $A$, which is at most $1/d(v)$. 
    The cost of making $(A, \alpha')$-poor nodes singletons is, at most, the sum of their degrees. 
    By \cref{lem:poornodes}, that expected sum of $(A, \alpha')$-poor nodes degrees is upper bounded by $\delta \cdot \EE[\pi]{cost_A}$.
    Hence, we have
    \[
        \EE[\pi]{cost_{A''}} \le \rb{1+ \delta} \EE[\pi]{cost_{A'}}.
    \]
    Finally, note that any $A$-lost node is a $A''$-light node or $A''$-heavy node with parameter $\beta/3$. 
    To see this, fix a ordering $\pi$. Consider a $A$-lost node $u$, and let $C=C_A^{\pi}(u)$. 
    Let $C'\subseteq C$ be the set of $A$-bad nodes in $C$. Let $C'':=C_{A''}^{\pi}(u)=C\setminus C'$. 
    By the definition of $A$-lost nodes, we have $d_{C'}(u)\ge \beta d_{C''}(u)$. If $u$ is not $A''$-light, then $d_{C''}(u)\ge |C''|/3$. So $d(u)\ge d_{C'}(u)\ge \beta d_{C''}(u)\ge \frac{\beta}{3}|C''|$. So $u$ is $A''$-heavy node with parameter $\beta/3$. By \cref{lem:heavynode}, $cost_{B}^\pi\le \frac{\beta+3}{\beta-3}cost_{A''}^{\pi}$. So $cost_{B}\le \frac{\beta+3}{\beta-3}cost_{A''}$. Putting all the steps together yields
    \[
        \EE[\pi]{cost_B} \le \frac{\beta+1}{\beta-1} \cdot \rb{1+\delta}\cdot \frac{\beta+3}{\beta-3} \cdot \EE[\pi]{cost_A}.
    \]
    Now since $\beta\ge \frac{4+\epsilon}{\epsilon}$, we have $\frac{\beta+1}{\beta-1}<\frac{\beta+3}{\beta-3}\le (1+\epsilon)$, and since $\alpha'<1/13$, we have $\frac{1+3/2 \cdot \alpha'}{1-5/2 \cdot \alpha'}<2$, and $\alpha'<\epsilon/8$ gives us $4\alpha'\frac{1+3/2 \cdot \alpha'}{1-5/2 \cdot \alpha'}<\epsilon$. So $\EE[\pi]{cost_B}\le (1+\epsilon)^3 \EE[\pi]{cost_A}\le (1+7\epsilon)\EE[\pi]{cost_A}$.
\end{proof}

\begin{lemma}\label{lem:bad-cluster-singleton-cost}
    Consider a fixed ordering $\pi$ and a clustering algorithm $A$. If $B$ is a clustering algorithm that runs $A$ and makes all $A$-bad and $A$-lost nodes, as well as all the nodes in bad clusters in $A$ singleton, then $cost(B)\le (1+8\epsilon) cost(A)$.
\end{lemma}
\begin{proof}
    Let $A'$ be the algorithm that runs $A$ and makes the $A$-bad nodes and $A$-lost nodes singleton. Then we can see $B$ as running $A'$ and making all the good nodes remaining in $A$-bad clusters singleton. Note that by \cref{lem:bad-and-lost-single} $cost(A')\le (1+7\epsilon) cost(A)$. Note that in $A'$ all poor clusters are singletons, since a poor cluster has at least one $(A,\alpha)$ poor node, and by \cref{lem:all-poor} all the non-pivot nodes in this poor cluster that are not heavy or light are $(A,3\alpha\beta)$-poor. 
    By definition of bad nodes, all the nodes in a poor cluster are bad.
    
    Consider a bad cluster $C$ in $A$. The cost of $C$ in $A$ is at least $\frac{2}{3}(1-\gamma)|C|^2$ by \cref{lem:bad-cluster-cost}. Making the good nodes in $C$ singleton in $A'$ adds at most $\gamma^2|C|^2$ to the cost of $A'$ since there are at most $\gamma|C|$ good nodes in $C$, and the only cost added by making these nodes singleton is through the edges between them. So the total cost added to the cost of $A'$ by making these good nodes singleton is at most $\frac{3\gamma^2}{2(1-\gamma)} cost(A)$. So $cost(B)\le cost(A')+\frac{3\gamma^2}{2(1-\gamma)}cost(A)$. Since $\gamma\le \epsilon/2$, we have $\frac{3\gamma^2}{2(1-\gamma)}\le \frac{4\gamma^2}{(1-\gamma)^2}\le \epsilon^2<\epsilon$. So $cost(B)\le (1+8\epsilon)cost(A)$. 
\end{proof}

%\mina{corrolary of above: if we make all the nodes of a bad cluster $C$ singleton, we loose a factor of $\frac{\gamma^2}{2(1-\gamma)/3}$ in the cost of that cluster, since the cost of making good nodes singleton is at most $\gamma^2|C|^2$.} 

\subsection{\ourPivot comparison to \refClustering}
Let $P_{alg}^\pi$ be the pivot set of our algorithm and $P_{ref}^\pi$ be the pivot set of the reference clustering.

\begin{lemma}\label{lem:pivot-set}
    Let $L'=L/2\beta$. If $v\in P_{ref}^\pi$, then $\pi(v)\le L'/d(v)$ holds with probability at least $1 - 1/n^{c-1}$.
    %Moreover, if $L = (c + 1) \cdot \ln n$ for sufficiently large constant $c$, then at any point in time, it holds that $P_{ref}^\pi \subseteq P_{alg}^\pi$ with high probability.
    %\snote{This second part of the claim does not seem to be used right now. Remember to either use it or remove it.}
\end{lemma}

Note that not only does \cref{lem:pivot-set} say that each pivot in \refClustering is also a pivot in \ourPivot, but it also provides a stronger guarantee on $\pi(v)$.

For the next Lemma, note that the definitions of light, poor, etc, are well-defined if the clustering is fixed (and not necessarily the ordering $\pi$). We show that with high probability, one of the good nodes in $C_{ref}(v)$ triggers \textsc{explore} function for $v$ so that $v$ can correct its cluster.
\begin{lemma}\label{lem:1-eps-good-nodes}
Fix a clustering and its pivots that \refClustering algorithm can produce. 
Consider a pivot $v$ whose cluster $C_{ref}(v)$ is \emph{good}.
Let $C_{ref}(v)[good]$ be the good nodes in $C_{ref}(v)$. 
Then, with high probability, our algorithm assigns $v$ as the pivot of $u$, for any $(1-\epsilon)|C_{ref}(v)[good]|$ first nodes $u$ of $C_{ref}(v)[good]$. The ``first'' here is taken with respect to the dynamic ordering and the probability taken over rankings that produce the fixed clustering. 
\end{lemma}

\cref{lem:misplaced} shows that \ourPivot identifies the pivot of all the good nodes that are not lost correctly.
\begin{lemma}\label{lem:misplaced}
    Fix a clustering and its pivots that the reference algorithm can produce. 
    Let $v$ be a pivot in that reference clustering and $u$ be in the $v$'s cluster. 
    With high probability, any $u$ that is ref-good and not ref-lost is assigned to $v$ by \cref{alg:node-dynamic-pivot}. 
\end{lemma}
\begin{proof}
    Let $C = \CCref(v)$.
    First, if $u$ is among the first $(1-\eps)$ fraction of the nodes of $C[good]$ with respect to the dynamic ordering, then by \cref{lem:1-eps-good-nodes}, with high probability \cref{alg:node-dynamic-pivot} assigns $u$ to $v$'s cluster.
    
    So, second, consider the case when $u$ is in the last $\eps$ fraction of $C[good]$.
    Let $D$ refer to the $ref$-good neighbors of $u$ that are among the first $1-\eps$ fraction of the nodes of $C[good]$. By \cref{lem:1-eps-good-nodes}, with high probability \cref{alg:node-dynamic-pivot} assigns the nodes in $D$ to $v$'s cluster. We now lower-bound $|D|$. 
    We will use that lower bound to show that the sampling process in \cref{alg:node-dynamic-pivot} (when $\pi(u)<L/d(u)$) will sample at least one node from $D$, and hence assign $v$ to $u$'s cluster.

    Since $u$ is not $ref$-lost, by \cref{def:light-poor-heavy-lost}, $u$ has at most $\beta$ times more $ref$-bad than $ref$-good neighbors in $C$. 
    Hence, at least $1/(1+\beta) \cdot |C|$ neighbors of $u$ in $C$ are $ref$-good.
    Also, observe that $1/(1+\beta) \cdot |C| - \eps \cdot |C[good]| \ge \rb{1/(1+\beta) - \eps} \cdot |C|$ of those neighbors are among the first $1-\eps$ fraction of $C[good]$. Therefore, $|D| \ge \rb{1/(1+\beta) - \eps} \cdot |C|$. On the other hand, since $u$ is a $ref$-good node, we have that $d(u) < \beta \cdot |C|$. 

    Finally, we conclude that
    \[
        \frac{|D|}{d_{insert}(u)}\ge \frac{|D|}{d(u)} \ge \frac{1/(1+\beta) - \eps}{\beta}.
    \]
    where $d_{insert}(u)$ is the degree of $u$ at the time of insertion. 
    Let $\bar{n}$ be the number of nodes when inserting $u$. If $x:=\frac{1/(1+\beta) - \eps}{\beta}$, the probability of not sampling any node in $D$ is at most $(1-x)^{|S|}$, where $S$ is the sample set. Recall that $|S|\ge 100\log(\frac{1}{1-x})\cdot \log \bar{n}$, so this probability is at most $1/\bar{n}^{100}$. Note that if $n$ is the current number of nodes, $\bar{n}\ge n/(1+\epsilon)$ since we recompute everything when the number of updates is at most $\eps$ times the number of nodes at last \recompute. So the probability that $u$ is clustered correctly is at most $1-(1+\epsilon)/n^{100}$.  %\mina{I prefer not to use the fact that we recompute things here, so that it works for the case where we only do insertions. But don't know how to avoid it.}
    %\snote{Update the sampling process to take into account this constant. Derive a condition on $\beta$ vs $\eps$.}
   % \mina{todo} because it cannot have too many bad neighbors since it is not lost, and so it has to have a lot of neighbors in the first $(1-\epsilon)|C_v[good]|$ nodes, so the sampling will hit $v$.
\end{proof}

Next, \cref{lem:4-approx-break-cluster} aids us to compare \refClustering with \ourPivot clustering. %Recall that the cost of a cluster is the number of non edges inside the cluster plus half of the number of edges between this cluster and other clusters. This way the total cost of a clustering is the sum of the cost of the clusters. 
Given pivot $v$ and set $B_v$, let $C_t$ be the set of nodes in $B_v$ with degree at least $t$. Let $cost(B_v|C_t)$ be the cost of the clustering on $B_v$ where all the nodes in $C_t$ are clustered as one cluster and the nodes in $B_v\setminus C_t$ as singletons. In particular, this cost equals half of the number of edges from $B_v$ to outside of $B_v$, plus the number of edges with at least one endpoint in $B_v\setminus C_t$, plus the number of non-edges in $C_t$.
\begin{lemma}\label{lem:4-approx-break-cluster}
    Consider a pivot $v$, and let $C^*$ be the set of good nodes that are not lost in $C_{ref}(v)$. There is a threshold $t\in \{1,(1+\epsilon),\ldots, (1+\epsilon)^{\ceil{\log n}+1}\}$, such that $cost(B_v|C_t)\le \frac{4}{1-2\eps}cost(B_v|C^*)$.% the cost of $C_t$ as a cluster and all the nodes in $B_v\setminus C_t$ as singletons is at most $4/(1-2\epsilon)$ times the cost of $C^*$ as one cluster with all the nodes in $B_v\setminus C^*$ as singletons.
\end{lemma}

% \begin{lemma}[Estimating the cluster cost]
%     Let $C \subseteq V$ be a set of nodes.
%     \cref{alg:within-cluster-cost-estimate} (\withinClusterEstimate) uses $O(|C| \log(n) / \eps^2)$ running time and with high probability either:
%     \begin{itemize}
%         \item Estimates, up to $1\pm \eps$ multiplicative factor, the number of non-edges within $C$, or
%         \item Correctly reports that the number of non-edges within $C$ is less than $3 \eps |C|$.
%     \end{itemize}
% \end{lemma}

\subsection{Putting it all together}

\textit{Proof of \cref{thm:approximation}.}
%\mina{good doesn't exclude lost so double check definitons and notations.}
We refer to \refClustering as $ref$ and the clustering of \ourPivot by $B$. Let $A$ be the clustering algorithm that runs \refClustering to obtain $ref$, and then makes all $ref$-bad,  $ref$-lost nodes as well as all the nodes in $ref$-bad clusters singleton. By \cref{lem:bad-cluster-singleton-cost}, $\mathbb{E}(cost(A))\le (1+8\epsilon)\mathbb{E}(cost(ref))$. Note that all the nodes that are not singletons in $A$ are good nodes that are not lost and are not in bad clusters in $ref$. 

% Let $B'$ be the algorithm that runs \ourPivot to obtain $B$, makes all $ref$-bad and $ref$-lost nodes singleton, and if any $ref$-good node was made singleton by \textsc{singleton-check}, it puts it back in its pivot's cluster. 

% First we show that $B'=A$. Consider any pivot $v$ in reference clustering. WHT, it is a pivot in $B$, and so in $B'$. The cluster containing $v$ in $A$ only contains the $ref$-good nodes of $C_{ref}(v)$. Moreover, by \cref{lem:misplaced} any $ref$-good node in $C_{ref}(v)$ is assigned to $v$ in $A$. Since in $B'$ all $ref$-bad and $ref$-lost nodes are made singleton, the cluster containing $v$ is the same in $A$ and $B'$. 

Next we show that $\mathbb{E}(cost(B))\le 4(1+O(\epsilon)) \mathbb{E}(cost(A))$. Consider a pivot $v$ in \refClustering. By \cref{lem:pivot-set} $v$ is also a pivot in \ourPivot with high probability. We will compare clustering costs by dividing up the clusters into groups: We consider all the nodes that are in $B_v$ for a pivot $v$ together and compare the cost of clustering these nodes in $A$ and in $B$.
We have two cases:
\paragraph{Case 1: $C_{ref}(v)$ is a good cluster.} Consider $B_v$, the set of nodes that are assigned to $v$ as their pivot, and $C_v\subseteq B_v$, the set of nodes in $B_v$ that are clustered with $v$ and are not singletons. 

Let $C_t=\{u\in B_v | d(u)\ge t \}$, and let $cost(B_v|C_t)$ be the cost of clustering all nodes in $C_t$ as one cluster and the nodes in $B_v\setminus C_t$ as singletons. %\footnote{In particular, this cost equals half of the number of edges from $B_v$ to outside of $B_v$, plus the number of edges with at least one endpoint in $B_v\setminus C_t$, plus the number of non-edges in $C_t$.}.
 Let ${t^*}$ be the threshold in $\{1,(1+\epsilon),\ldots,(1+\epsilon)^{\ceil{\log n}+1}\}$ where $cost(B_v|C_{t^*})$ is minimized. %the clustering with $C_{t^*}$ as one cluster and all $B_v\setminus C_{t^*}$ singleton has the lowest cost among all $t$. 
By \cref{thm:cost-estimate}, $cost(B_v|C_v)\le (1+220\eps)cost(B_v|C_{t^*})$.

%, then recall that $C_v=C_{t_v}$ where $t_v$ is the threshold in $\{1,(1+\epsilon),\ldots,(1+\epsilon)^{\log n}\}$ where the clustering with $C_t$ as one cluster and all $B_v\setminus C_t$ singleton has the lowest cost among all $t$. 
If $C^*_v$ is the set of good nodes that are not lost in $C_{ref}(v)$, then by \cref{lem:misplaced} all the nodes in $C_v^*$ are correctly assigned to $v$, and so they are in $B_v$.
Note that $A$ clusters the node in $B_v$ as follows: put all the nodes in $C_v^*$ in one cluster and make all the nodes in $B_v\setminus C_v$ singleton. 
By \cref{lem:4-approx-break-cluster}, we have that $cost(B_v|C_{t^*})\le \frac{4}{1-2\eps}cost(B_v|C_v^*)$, so $cost(B_v|C_{v})\le \frac{4(1+220\eps)}{(1-2\eps)}cost(B_v|C_v^*)\le 4(1+230\eps)cost(B_v|C_v^{*})$ since $\eps\le 1/1000$. %this clustering has cost at least $(1-2\epsilon)/4$ of the cost of the clustering that $B$ does on $B_v$, namely making $C_v$ one cluster and all the nodes in $B_v\setminus C_v$ singleton. \mina{this paragraph may need adjusting based on cluster estimation}

\paragraph{Case 2: $C_{ref}(v)$ is a bad or poor cluster.} We know that all the nodes in $C_{ref}(v)$ are singletons in $A$. Let $t^*$ be the threshold in $\{1,(1+\epsilon),\ldots,(1+\epsilon)^{\ceil{\log n}+1}\}$ where $cost(B_v|C_{t^*})$ is minimized. So $cost(B_{t^*}|C_v)$ is at most the cost of making $B_v$ singleton, i.e. $cost(B_v|C_{(1+\epsilon)^{\ceil{\log n}+1}})=cost(B_v|\emptyset)$. By \cref{thm:cost-estimate} $cost(B_v|C_v)\le (1+220\eps)cost(B_v|C_{t^*})$. So $cost(B_v|C_v)\le (1+220\eps)cost(B_v|\emptyset)$.% the cost of making all  Now note that by \mina{estimating clusters, might need more expplanation..}, the cost of clusters $C_v$ and singletons $B_v\setminus C_v$ is at most $(1+\epsilon)$ the cost of $B_v$ being singleton. 

Finally, note that a node not in any $B_v$ is a singleton in both $B$ and $A$. 
Putting the above two cases together, 
So we have $\mathbb{E}(cost(B))\le 4(1+50\epsilon) \mathbb{E}(cost(A))\le 4(1+230\epsilon)(1+8\epsilon)\mathbb{E}(cost(ref))\le 4(1+1000\epsilon)\mathbb{E}(cost(ref))$, where the last inequality uses the fact that $\epsilon<1/1000$.

\section{Random deletions}
Let $n_0$ be the number of nodes in the graph just after the last recomputation.
Our algorithm for deletions is quite simple:

\noindent 1. Ignore deletions.

\noindent 2. After $\eps n_0/6$ updates, counting both insertions and deletions, recompute the clustering from scratch by treating all non-deleted nodes as if they had been inserted one by one again. 
    These insertions are processed by \cref{alg:node-dynamic-pivot}.
    Our recomputing procedure is described in \cref{sec:recompute}.

Recall that on a deletion update, a node to be deleted is chosen uniformly at random among the existing nodes.
Observe that the choice of deletions is independent of the randomness used by our algorithm.
At time $t$ of the algorithm, let $D_t$ be the nodes deleted since the last clustering recomputation. 
We think of $D_t$ as the nodes waiting to be deleted.
%So, when at a time $t$ we want to argue about the approximation of our algorithm, we can also assume that $D_t$ nodes have all arrived after the last insertion. 

By construction, we have $|D_t| \le \eps n_0 /6$. 
% So, instead of assuming $D_t$ contains at most $\eps n_0 / 6$ nodes, we assume $D_t$ is constructed by sampling each node from $V$ to it independently with probability (at most) $\eps$.
% This different sampling process is very close to the original one, but it is easier to analyze the algorithm assuming nodes are sampled into $D_t$ independently. 
% We discuss how to adapt the analysis of this different sampling process to the original one in \cref{sec:different-original-deletion-process}.
Since after the recomputation there are $n_0$ nodes in the graph, $D_t$ is a subset of (at least) $n_0$ nodes, and hence $\prob{u \in D_t} \le |D_t| / n_0 \le \eps / 6$. There is inequality instead of equality, as the $\eps n_0/6$ updates might contain insertions, resulting in a reduced probability of a node appearing in $D_t$.

Let $\cCnodel$ be the clustering obtained by $5$-approximate \refClustering at time $t$ where deletions $D_t$ are ignored. 
Let $\cCdel$ be the clustering obtained by \refClustering at time $t$ in which deletions $D_t$ are considered.
Let $P_t \subseteq V \times V$ be the node pairs that $\cCnodel$ pays for. 
We aim to lower-bound the expected number of pairs in $P_t$ that $\cCdel$ also pays for.
Consider a pair $e = \{u, v\} \in P_t$. We will lower-bound the probability that the clustering of $u$ and $v$ is the same in $\cCdel$ as in $\cCnodel$

Consider a node $u$; the exact same analysis applies to $v$. 

First, assume that $u$ is a singleton in $\cCnodel$. 
This implies that the neighbor of $u$ with smallest rank, $w$, is not a pivot. 
    Node $w$ is not a pivot because it has a neighbor $w'$ with a rank smaller than $w$. 
    If none of $w, w'$, or $u$ is in $D_t$, then $u$ is a singleton in $\cCdel$. 
    Since we have $\prob{w \in D_t \text{ or } w'  \in D_t \text{ or }  u \in D_t} \le \prob{w \in D_t} + \prob{w' \in D_t} + \prob{u \in D_t} \le \eps/2$, then in this case, the clustering of $u$ in $\cCdel$ is the same as in $\cCnodel$ with probability at least $1 - \eps/2$.

Second, assume that $u$ is not a singleton in $\cCnodel$.
    This implies that the highest-rank neighbor $w$ of $u$ is a pivot. 
    By deleting nodes, and unless $w$ is deleted, $w$ remains a pivot. So, unless $u$ or $w$ are in $D_t$, the clustering of $u$ is the same in $\cCdel$ and $\cCnodel$.
    Since we have $\prob{w \in D_t \text{ or } u \in D_t} \le \eps/3$, in this case, the clustering of $u$ in $\cCdel$ is the same as in $\cCnodel$ with probability at least $1 - \eps/3$.

This analysis implies that the clustering of a pair $\{u, v\}$ is the same in $\cCdel$ as in $\cCnodel$ with probability at least $1 - 5\eps/6$. Hence, by the linearity of expectation,
\[
    \EE{\cost(\cCdel)} \ge (1 - 5\eps/6) |P_t| \ge (1 - \eps) \cost(\cCnodel).
\]

\section{Experiments}\label{sec:experiments}
In this section, we empirically demonstrate that our approximation guarantee is better than that of \refClustering and the algorithm of \cite{cohendynamic}. In the rest, we use \textsc{dynamic agreement} to refer to the approach in \cite{cohendynamic}. 
%\subsection{Set-up}
\paragraph{Algorithm Parameters}In \ourPivot, for the \textsc{Break-cluster} and \textsc{update-cluster} subroutines we do the following: in \textsc{Break-cluster} we consider $O(\log{n})$ many candidates for $C_v$, estimate their costs and pick one with the lowest cost. In \textsc{update-cluster} we update our $O(\log{n})$ estimates by adding the new node, and again pick the one with lowest cost. 
To simplify the code, we heuristically alter the \textsc{Break-cluster} and \textsc{update-cluster} subroutines as follows.
In \textsc{Break-cluster}, %instead of estimating the cost $O(\log n)$ many different ways to break $B_v$, 
for each node $u\in B_v$, we sample $O(\log n)$ nodes in $B_v$. If $u$ is attached to half of them, we add $u$ in $C_v$.
In \textsc{update-cluster}, we add $u$ to $C_v$, even though $u$ might not be attached to many nodes in $C_v$. After at least $\eps|B_v|$ nodes are added to $B_v$, we rerun \textsc{Break-cluster}.
Note that the reason we get a $(20+O(\eps))$ approximation instead of a $(5+O(\eps))$ approximation is the \textsc{Break-cluster} subroutine, so depending on the application, one can replace this subroutine with a version that one sees fit. Furthermore, we run \textsc{recompute} every time the number of deletions reaches $\eps N$, instead of the total number of updates. We observe that this does not degrade the approximation guarantee and slightly improves the running time. 

We set the experiment parameters and the parameters for \textsc{dynamic-agreement} to be the same as in \cite{cohendynamic}. We choose a random ordering for the arrival of the nodes, and at each step, with probability $0.8$, we insert the next node, and with probability $0.2$, we delete a random node. 
If all the nodes have been inserted once, we delete them until no node is left. We set the parameter $\eps$ for \ourPivot to be $0.1$.

\noindent\textbf{Datasets}
We use the same datasets as in \cite{cohendynamic} for a complete comparison. 
We evaluate the algorithms on two types of graphs. 
\\
(1) Sparse real-world graphs from SNAP \cite{jure2014snap}: a social network (musae-facebook), an email network (email-Enron), a collaboration network (ca-AstroPh), and a paper citation network (cit-HepTh)
\\
(2) The drift dataset \cite{vergara2012chemical,rodriguez2014calibration} from ICO Machine Learning Repository \cite{dua2017uci}, which includes 13,910 points embedded in a space of 129 dimensions. A graph is constructed by placing an edge between two nodes if their Euclidean distance is less than a certain threshold. This setup is used to easily change the density of the graph and test how it affects the algorithms. The thresholds we choose are the same as in \cite{cohendynamic}, and they are the mean of the distances between all nodes divided by $c\in \{10,15,20,25,30\}$. The lower the threshold, the sparser the graph. The density of a graph is the ratio of the number of edges and the number of nodes.

\noindent\textbf{Baselines} We use three baselines: making all nodes singletons, which we call \textsc{singletons}, \textsc{dynamic-agreement}, and \refClustering. We divide the cost of each algorithm by the cost of \textsc{singletons}. 
Since \refClustering handles only node insertions, we process deletions in a way similar to \cite{cohendynamic}. 
Note that our results \emph{slightly} differ from that of \cite{cohendynamic} since they depend on the randomness of node arrivals.
Moreover, the running time depends on the machine in which the algorithm is being run.
Nevertheless, the scale of results we obtain does reproduce that of \cite{cohendynamic}.
%\stodo{Let me know if the last sentence is incorrect.}it ok

%\subsection{Results}
\paragraph{Results: Approximation Guarantee}
For all the datasets, our approximation guarantee is better than \textsc{dynamic-agreement} and \textsc{singletons}. 
For SNAP graphs, we plot the correlation clustering objective every 50 steps. \cref{fig:email-enron} shows this objective for one of these graphs, and the rest can be found in \cref{sec:extra-exp}. 
\begin{figure}[h]
    \centering
    \includegraphics[width=0.8\linewidth]{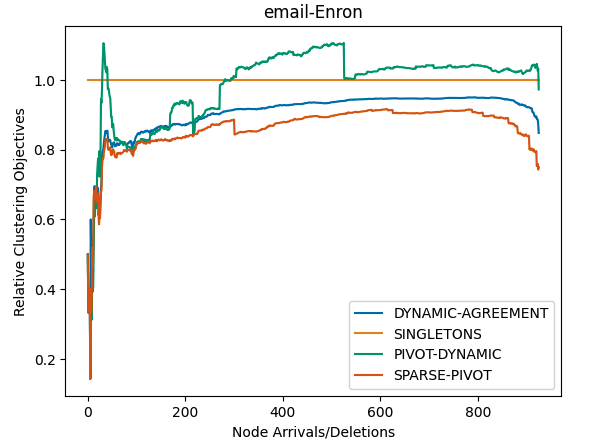}
    \caption{Comparison of the correlation clustering objective across the four algorithms. The lower the plot, the better.}
    \label{fig:email-enron}
\end{figure}
The average clustering objective for the drift dataset graphs is shown in \cref{tab:drift-cc-obj}.

\begin{table}[]
    \centering
    \begin{tabular}{c|c|c|c}
        Density & DA  & RC  & SP  \\ \hline
        $235.36$ & $0.69$ & $0.59$ & $0.6$\\ 
        $114.87$ & $0.59$ & $0.64$ & $0.49$ \\
        $69.74$ & $0.5$ & $0.5$ & $0.41$ \\
        $52.17$ & $0.39$ & $0.42$ & $0.32$ \\
        $42.25$ & $0.35$ & $0.35$ & $0.29$ \\ \hline
    \end{tabular}
    \caption{Clustering Objective of \textsc{dynamic-agreement} (DA), \refClustering (RF) and \ourPivot (SP). The smaller the number, the better.}
    \label{tab:drift-cc-obj}
\end{table}

\noindent\textbf{Results: Running time} Our experiments focus on the solution quality of \ourPivot.
Nevertheless, we compare the running times for completeness and illustrate that \ourPivot is faster than \textsc{dynamic-agreement} in practice, see \cref{sec:extra-exp}.

\section*{Acknowledgements}
K.~Makarychev was supported by the NSF Awards CCF-1955351 and EECS-2216970.

S.~Mitrovi\' c was supported by the NSF Early Career Program No.~2340048 and the Google Research Scholar Program.

\section*{Impact Statement}
This paper presents work that aims to advance algorithmic tools for data partitioning, a method used in the field of Machine Learning. 
There are many potential societal consequences of our work, none of which we feel must be specifically highlighted here.

\bibliography{references}
\bibliographystyle{icml2025}

%%%%%%%%%%%%%%%%%%%%%%%%%%%%%%%%%%%%%%%%%%%%%%%%%%%%%%%%%%%%%%%%%%%%%%%%%%%%%%%
%%%%%%%%%%%%%%%%%%%%%%%%%%%%%%%%%%%%%%%%%%%%%%%%%%%%%%%%%%%%%%%%%%%%%%%%%%%%%%%
% APPENDIX
%%%%%%%%%%%%%%%%%%%%%%%%%%%%%%%%%%%%%%%%%%%%%%%%%%%%%%%%%%%%%%%%%%%%%%%%%%%%%%%
%%%%%%%%%%%%%%%%%%%%%%%%%%%%%%%%%%%%%%%%%%%%%%%%%%%%%%%%%%%%%%%%%%%%%%%%%%%%%%%
\newpage
\appendix
\onecolumn
\section{Implementing cost estimates}\label{sec:cost-estimate}
Our clustering procedures, e.g., \cref{alg:node-dynamic-pivot}, for each pivot $v$ maintain two sets of nodes: $B_v$ and $C_v$. The set $B_v$ is a set of nodes whose pivot is $v$.
However, having $B_v$ as one cluster might sometimes be very far from an optimal clustering of the nodes within $B_v$.
So, our algorithm computes a cluster $C_v \subseteq B_v$ for which we can guarantee a relatively low cost; details of this analysis are provided in our proof of \cref{thm:approximation}.
To compute a cluster $C_v$, our algorithm estimates the costs of several clusters and chooses $C_v$ as the cluster with the lowest estimated cost.
In this section, we describe how to estimate the cost of a cluster efficiently, that is, in only $\poly(\log n, 1/\eps)$ time.
We need to handle two cases: how to estimate the cost of a given cluster $C$ from scratch, i.e., in a static manner, and how to maintain the cost estimate of a cluster $C$ under node insertions.

We need the former case for our recomputation or when we create an entirely new $B_v$ because $v$ is just becoming a pivot.
It might be tempting to create new $B_v$ by ``pretending'' that the nodes of $B_v$ have been inserted one by one. 
However, this approach has a small subtlety. 
Namely, when a node $v$ is inserted, only at that point are the edges incident to $v$ included in our graph, and no edge between $v$ and a node inserted in the future is known. 
On the other hand, if we ``pretend'' that an already existing sequence of nodes is just now inserted, then a currently processed node also has edges to its neighbors that have yet to be processed/inserted.
This scenario slightly affects how we count edges and non-edges within $B_v$ or $C_v$.

\subsection{Static version}

\subsubsection{Within-cluster cost estimate}
We first design a procedure to estimate the cost within a cluster $C$, i.e., the number of non-edges within $C$, by spending only $O(\log n)$ time per edge. 
It is provided as \cref{alg:within-cluster-cost-estimate}.
\begin{algorithm}[h]
\caption{\withinClusterEstimate \label{alg:within-cluster-cost-estimate}}
\begin{algorithmic}[1]
    \STATE \textbf{Input} set $C\subseteq V$
    \STATE $\tau_C \eqdef 5 \cdot |C| \cdot \log(n) / \eps^3$
    \STATE \textbf{for} $i = 1 \ldots \tau_C$:
    \STATE \ourIndent[1] Uniformly at random, sample two distinct nodes $v$ and $w$ from $C$
    \STATE \ourIndent[1] \textbf{if } $\{w, v\}$ is a non-edge, then $S \gets S + 1$
    \STATE \textbf{return} $S \cdot \binom{|C|}{2} / \tau_C$
    % \STATE \textbf{if } $Y \ge \eps n$
    % \STATE \ourIndent[1] \textbf{return } $Y$
    % \STATE \textbf{else} 
    % \STATE \ourIndent[1] \textbf{return } $0$
\end{algorithmic}
\end{algorithm}
As shown by \cref{lemma:stating-in-cluster-cost}, this estimate is tightly concentrated as long as the number of non-edges is in $\Omega(|C|)$. If the number of non-edges is lower, then their actual number is irrelevant to our algorithm.
Updating this cost dynamically is more involved, and we elaborate on details in \cref{sec:dynamic-cost-estimate}.

\begin{lemma}[In-cluster cost estimate]\label{lemma:stating-in-cluster-cost}
    Let $C \subseteq V$ be a set of nodes.
    Then, for $\eps < 1/2$, \cref{alg:within-cluster-cost-estimate} (\withinClusterEstimate) uses $O(|C| \log(n) / \eps^2)$ running time and outputs $Y$ for which with high probability the following holds:
    \begin{itemize}
        \item If the number of non-edges within $C$ if at least $2 \eps |C|$, then $Y$ is a $1\pm \eps$ multiplicative approximation of that number of non-edges.
        \item Otherwise, $Y < 3 \eps |C|$.
    \end{itemize}
\end{lemma}
\begin{proof}
    Let $t$ be the number of non-edges in $C$.
    Let $X_i$ be a random $0/1$ variable equal $1$ iff the $i$-th $\{v, w\}$ pair sampled by \cref{alg:within-cluster-cost-estimate} is a non-edge. 
    Then,
    \[
        \EE{X_i} = \prob{X_i = 1} = \frac{t}{\binom{|C|}{2}}.
    \]
    Let $S'$ be the value of $S$ at the end of \cref{alg:within-cluster-cost-estimate}. 
    Since $S' = \sum_{i = 1}^{\tau_C} X_i$, we have
    \begin{equation}\label{eq:EE-S'}
        \EE{S'} = \tau_C \cdot \frac{t}{\binom{|C|}{2}}.
    \end{equation}
    Let
    \[
        Y \eqdef \frac{S' \cdot \binom{|C|}{2}}{\tau_C}
    \]
    be the output of \cref{alg:within-cluster-cost-estimate}. 
    Observe that $\EE{Y} = t$.
    Therefore, the expected value of $Y$ is the desired one. 
    In the rest, we analyze the concentration bounds of this estimator.

    Consider two cases based on the value of $t$.
    \paragraph{Case $t \ge 2 \eps |C|$.}
    Recall that $\tau_C = 5 |C| \log(n) / \eps^3$.
    Replacing the bounds on $t$ and $\tau_C$ in \cref{eq:EE-S'} yields
    \begin{equation}\label{eq:EE-S'-uuper-bound}
        \EE{S'} \ge \frac{10 |C|^2 \cdot \log n}{\eps^2 \binom{|C|}{2}} \ge \frac{20 \log n}{\eps^2}.
    \end{equation}
    Since $S'$ is a sum of independent $0/1$ random variables, by the Chernoff bound, it holds that\footnote{The constant $-6$ can be made arbitrarily large by increasing the constant in $\tau_C$.} 
    \[
        \prob{|S' - \EE{S'}| < \eps \EE{S'}} \le n^{-6}.
    \]
    This now implies that for $t \ge 2 \eps |C|$, with probability at least $1 - n^{-6}$, $Y$ is a $(1\pm \eps)$ multiplicative approximation of $t$. 
    
    \paragraph{Case $t < 2 \eps |C|$.}
    In this case, we would like to claim that very likely it holds that $Y < 3 \eps |C|$.
    This can be argued by applying the Chernoff bound as follows.
    
    Observe that for this value of $t$ we have
    \[
        \EE{S'} < \frac{20 \log n}{\eps^2}.
    \]
    Hence, 
    \[
        \prob{S' > (1+\eps) \frac{20 \log n}{\eps^2}} \le n^{-6}.
    \]
    Therefore, with probability at least $1 - n^{-6}$, it holds that
    \[
        Y \le \frac{(1 + \eps)\frac{20 \log n}{\eps^2} \cdot \binom{|C|}{2}}{\tau_C} < (1+\eps) 2 |C| < 3 \eps |C|,
    \]
    for $\eps < 1/2$.
\end{proof}

\subsubsection{Single-cluster + singletons cost estimate}
We now discuss how to estimate the cost of $B$ for a given $C$, where $C$ is taken as a single cluster, while all the nodes in $B - C$ are singletons. By cost of $B$ we mean the number of edges with at least one endpint in $B-C$ and the other endpoint in $B$, plus the number of non-edges in $C$. Note that the true correlation clustering cost of $B$ is the above cost plus half of the edges from $B$ to outside of $B$, but since we need these costs to compare different choices of $B$ and the number of edges going outside of $B$ is indipendent of this choice, it does not influence our comparison.

In this section, given two node subsets $X$ and $Y$, we use $e(X, Y)$ to denote the number of edges with one endpoint in $X$ and the other in $Y$. In particular, $e(X, X)$ is the number of edges in $G[X]$. We abbreviate $e(X,X)$ to $e(X)$.

First, the entire cost of $C$, denoted by $cost(C)$, equals the sum of $e(C, V-C)$ and the number of non-edges within $C$. Observe that
\[
    e(C) = \binom{|C|}{2} - \text{[the number of non-edges within $C$]}.
\]
So, we have
\[
    \sum_{w \in C} d(w) = e(C, V - C) + 2 e(C) = e(C, V - C) + 2 \binom{|C|}{2} - 2 \cdot \text{[the number of non-edges within $C$]}.
\]
This now implies that
\[
    cost(C) = \sum_{w \in C} d(w) - 2 \binom{|C|}{2} + 3 \cdot \text{[the number of non-edges within $C$]}.
\]

Second, it remains to account for making the nodes in $B - C$ singletons. 
The edges $E(C, B - C)$ are already accounted for by $\sum_{w \in C} d(w)$.
To account for the cost of making $B-C$ singletons, the following procedure can be used:
\begin{itemize}
    \item For each node $w \in B-C$, iterate over all the edges in its adjacency list and
    \item to an edge from $E(w, C)$ assign weight $0$; to an edge from $E(w, B-C)$ assign weight $1/2$; and to an edge from $E(w, V-B)$ assign weight $1$.
\end{itemize}
Observe that an edge $\{x, y\}$ with $x, y \in B-C$ is counted twice: once in the adjacency list of $x$ and once in the adjacency list of $y$.
Hence, the sum of these edge weights and $cost(C)$ equals the cost of clustering $B$.

However, using the above procedure directly can result in a running time that is too long. 
Instead, we would like a procedure with the running time of $O(|B| \cdot \poly(\log n, 1/\eps))$.
Nevertheless, estimating the sum of the edge weights in the desired time is simple. We outline one such approach in \cref{alg:cost-estimate}.

\begin{algorithm}[h]
\caption{\costEstimate \label{alg:cost-estimate}}
\begin{algorithmic}[1]
    \STATE \textbf{Input} node sets $B$ and $C\subseteq B$.
    \STATE $\tcost = \sum_{w \in C} d(w) - 2 \binom{|C|}{2} + 3 \cdot \withinClusterEstimate(C)$ \label{line:tcost-C}
    \STATE Let $\eta \eqdef 10 \cdot \log(n) / \eps^3$
    \STATE \textbf{for } $w \in B - C$:
    \STATE \ourIndent[1] Sample $\eta$ edges incident to $w$, each edge sampled independently and uniformly at random
    \STATE \ourIndent[1] For $S \in \{C, B-C, V-B\}$, let $Z_S(w)$ be $d(w) / \eta$ multiplied by number of sampled edges incident to $S$
    \STATE \ourIndent[1] $\tcost = \tcost + \frac{Z_{B-C}(w)}{2} + Z_{V-B}(w)$
    \STATE \textbf{return } $(\tcost+9\eps|C|)/(1-37\eps)$
\end{algorithmic}
\end{algorithm}

Let $Z_S(w)$ be as defined in \cref{alg:cost-estimate}.
Observe that $\EE{Z_S(w)} = e(w, S)$.
A straightforward analysis, and identical to that presented in the proof of \cref{lemma:stating-in-cluster-cost}, shows that for $e(w, S) \ge 2 \eps d(w)$ the value of $Z_S(w)$ computed in \cref{alg:cost-estimate} is with high probability a $(1\pm \eps)$ factor approximation of $e(w, S)$. 
For $e(w, S) < 2\eps d(w)$, the same analysis yields that with high probability, it holds that $Z_S(w) < 3 \eps d(w)$.
Hence, when $e(w, S) < 2\eps d(w)$, the estimate $Z_S(w)$ is not necessarily within $1\pm \eps$ factor of its expected value, and thus the error has to be accounted for differently. Next, we explain how to account for it.

Trivially, at least one among $e(w, C)$, $e(w, B-C)$, and $e(w, V - B)$ is at least $d(w) / 3 > 2 \eps d(w)$, for $\eps < 1/6$.
If $e(w, C) \ge d(w) / 3$, we charge each $Z_S(w) < 3\eps d(w)$ to the cost of $E(w, C)$ paid by $\sum_{w \in C} d(w)$. This incurrs an extra cost of at most $(2 \cdot 3\eps d(w)) / (d(w) / 3) = 18 \eps$ per an edge in $E(w, C)$.
The analogous analysis applies to the case $e(w, V-B) \ge d(w) / 3$ and $e(w, B-C) \ge d(w) / 3$. The only difference is that $Z_{B-C}(w)$ is divided by $2$ in \cref{alg:cost-estimate}, so for that case, the analysis yields a $36 \eps$ increased cost per edge.

Overall, this analysis implies that, with high probability, $\tcost$ in \costEstimate is a $(1 \pm 37 \eps)$ multiplicative and $9 \eps |C|$ additive approximation of the cost of clustering $B$. 
The additive approximation comes from \cref{lemma:stating-in-cluster-cost} and the fact that $3 \cdot \withinClusterEstimate(C)$ figures in the output of $\costEstimate$.

\begin{lemma}[Cost estimate of single-cluster + singletons]
\label{lemma:cost-estimate}
    Let $\eps<1/111$. Given two node sets $B$ and $C \subseteq B$, let $cost^*(B|C)$ be the cost of clustering $B$ in which $C$ is a single cluster and $B-C$ are singletons, which is defined to be the number non-edges in $C$ plus the number of edges in $B$ with at least one endpoint in $B-C$.
    Then, if $\costEstimate(B, C)$ (\cref{alg:cost-estimate}) outputs $X$, we have $cost^*(B|C)\le X\le (1+111\eps) cost^*(B|C)+27\eps|C|$. %outputs $1+37\eps$ multiplicative and $9 \eps |C|$ additive approximation of $cost(B|C)$.
    Moreover, the algorithms run in $O(|B| \cdot \log(n) / \eps^3)$ time. 
\end{lemma}
Note that since $(1-37\eps)cost^*(B|C)-9\eps|C|\le\tcost \le (1+37\eps)cost^*(B|C)+9\eps|C|$ and $X=(\tcost+9\eps|C|)/(1-37\eps)$, and $\eps<1/111$ we have that $cost^*(B|C)\le X\le (1+111\eps) cost^*(B|C)+27\eps|C|$.

\subsubsection{Cost Comparison}
%In this section, given two node subsets $X$ and $Y$, we use $e(X, Y)$ to denote the number of edges with one endpoint in $X$ and the other in $Y$. 
%In particular, $e(X, X)$ is the number of edges in $G[X]$, and we denote it by $e(X)$ for abbreviation. 
%Let $\bar{e}(X)$ denote the number of non-edges in $G[X]$.

%Given a set $B$ and its subset $C \subseteq B$, let $cost(B|C)$ denote the cost of clustering where $C$ is a single cluster and $B\setminus C$ are singletons. The value $cost(B|C)$) is the cost inside $B$, plus half the number of edges from $B$ to outside of $B$. 
%The cost inside $B$ is the number of edges with at least one endpoint in $B\setminus C$, plus the number of non-edges inside $C$. To compute this cost, we show that we only need the number of non-edges in $C$ and the degree of each node in $B$.  
% First, we have
% $$
%     cost(B|C) = \frac{1}{2} e(V\setminus B,B)+ e(B\setminus C,B)+\bar{e}(C)
% $$

% We can write the sum of degrees of the nodes in $B$ as follows:
% \begin{align*}
%     \sum_{w\in B}d(w) &= e(V\setminus B,B)+2e(B) \\
%     & = e(V\setminus B,B) + 2e(B\setminus C, B)+ 2e(C)\\
%     & = e(V\setminus B,B) + 2e(B\setminus C, B) + 2{|C|\choose 2} - 2\bar{e}(C)
% \end{align*}
% So
% \begin{equation}\label{eq:cost-B-C}
% cost(B|C) = \frac{1}{2}\rb{\sum_{w\in B}d(w)} -{|C| \choose 2}+2\bar{e}(C)
% \end{equation}

Let the estimate that \cref{alg:cost-estimate} makes for $cost^*(B|C)$ be $\tcost(B|C)$. We show that $\tcost(B|C)$ is a good enough measure for choosing a $C$ with low $cost(^*B|C)$.

\begin{algorithm}[h]
\caption{\textsc{cost-comparison} \label{alg:cost-comparison}}
\begin{algorithmic}[1]
    \STATE  Use \cref{alg:cost-estimate} to compute $\tcost(B|C)$ and $\tcost(B|C')$.
    \STATE \ourIndent \textbf{if} $\tcost(B|C)<\tcost(B|C')$:
    \STATE \ourIndent[2] \textbf{return} $C$
    \STATE \ourIndent \textbf{else}:
    \STATE \ourIndent[2] \textbf{return} $C'$.
\end{algorithmic}
\end{algorithm}

\begin{lemma}\label{lem:big-cost}
    If $cost^*(B|C)\ge |C|/4$ then $cost^*(B|C)\le \tcost(B|C)\le (1+219\eps) cost^*(B|C)$. 
\end{lemma}
\begin{proof}
    Since $|C|\ge 4cost(B|C)$, by \cref{lemma:cost-estimate}, we have that $\tcost(B|C)\le (1+111\eps)cost^*(B|C)+27\eps|C|\le (1+111\eps)cost^*(B|C)+108\eps cost^*(B|C)\le (1+219\eps)cost^*(B|C)$.
\end{proof}

\begin{lemma}\label{lem:low-cost}
If $cost^*(B|C)\le |C|/4$, then for any $C'$ such that $C'\subset C$ or $C\subset C'$, we have $cost^*(B|C)<cost^*(B|C')$ and $\tcost(B|C)\le \tcost(B|C')$.
\end{lemma}
\begin{proof}
    First suppose that $C'\subset C$. Take a node $u\in C\setminus C'$, we know that $cost(B|C)\ge |C|-1-d_C(u)\ge |C|/2 - d_C(u)$. So $|C|/2\le d_C(u)$. Now we have $cost^*(B|C')\ge d_C(u)\ge |C|/2>cost^*(B|C)$. Similarly, suppose $C\subset C'$. Take a node $u\in C'\setminus C$. We know that $cost^*(B|C)\ge d_C(u)$, so $d_C(u)\le |C|/4$. Moreover, $cost^*(B|C')\ge |C|-1-d_C(u)> |C|/2$. So in both cases $cost^*(B|C')>|C|/2>cost^*(B|C)$.
    
    Furthermore, by \cref{lemma:cost-estimate} $\tcost(B|C)\le (1+111\eps)|C|/4+27\eps|C|$ and by \cref{lem:big-cost} we have $\tcost(B|C')\ge (1+219\eps)cost^*(B|C)\ge (1+219\eps)|C|/2$. So we have $\tcost(B|C')<\tcost(B|C)$.  
\end{proof}

We use \cref{alg:cost-comparison} to develop \cref{alg:break-cluster} that finds a dense cluster $C_v$ inside $B_v$, where $B_v$ is the set of all the nodes that are assigned to pivot $v$. Recall that $C_t$ is the set of nodes in $B_v$ with degree at least $t$.
\begin{algorithm}[h]
\caption{\textsc{break-cluster} \label{alg:break-cluster}}
\begin{algorithmic}[1]
    \STATE \textbf{Input} Set $B_v$.%, estimates $\tcost(B_v|C_t)$ for $t\in \{1,(1+\epsilon), (1+\epsilon)^2,\ldots,(1+\epsilon)^{\log n}\}$, where $C_t = \{u\in B_v, d(u) \ge t\}$.
    \STATE For any $t>0$, let $C_t = \{u\in B_v, d(u) \ge t\}$.
    \STATE initialize $t_v=0,C_v = B_v$.
    \STATE \textbf{for} $t\in \{1,(1+\epsilon), (1+\epsilon)^2,\ldots,(1+\epsilon)^{\ceil{\log n}}\}$: 
    \STATE \ourIndent[1] $C_v\gets$\textsc{cost-comparison}($B_v,C_t,C_v$).
    \STATE \textbf{return} $C_v$.
    \end{algorithmic} 
\end{algorithm}

% \begin{lemma}\label{lem:comparison-special-case}
%     Let $t'>t$, and suppose that $\tilde{\bar{e}}(C_{t'})\le 3\eps |C_{t'}|$ and $cost(B_v|C_{t'})\le |C_{t'}|/4$. Suppose that $C_{t'}\neq C_t$ and $\bar{e}(C_t)\ge 2\eps |C_t|$. Then we have $cost(B_v|C_t)>cost(B_v|C_{t'})$ and $\tcost(B_v|C_t)>\tcost(B_v|C_{t'})$.
% \end{lemma}
% \begin{proof}
%     Let $u\in C_t\setminus C_{t'}$. Since $u$ is not in $C_{t'}$, we have that $cost(B_v|C_{t'})$ is at least $d(u)/2$. So $d(u)/2\le |C_{t'}|/4$, thus $d(u)\le |C_{t'}|/2$. Now $cost(B_v|C_t)$ is at least the number of non-edges that $u$ and $C_{t'}$ have, which is at least $|C_{t'}|-d(u)\ge \frac{1}{2}|C_{t'}|$. So $cost(B_v|C_t)\ge \frac{1}{2}|C_{t'}|$. So $cost(B_v|C_t)>cost(B_v|C_{t'})$.

%     Now note that on the one hand since $\bar{e}(C_t)\ge 2\eps|C_t|$, by \cref{lem:high-non-edge} $\tcost(B_v|C_t)\ge (1-2\epsilon) cost(B_v|C_t)\ge \frac{(1-2\eps)}{2}|C_{t'}|$. On the other hand, $\tcost(B_v|C_{t'
%     }) = cost(B_v|C_{t'})-\bar{e}(C_{t'})+\tilde{\bar{e}}(C_{t'})\le |C_{t'}|/4+3\eps |C_{t'}| = \frac{1+12\eps}{4}|C_{t'}|$. Since $\eps<1/16$, we have that $\tcost(B_v|C_{t'})<\tcost(B_v|C_t)$.
% \end{proof}

For any $C\subseteq B$, let $cost(B|C)$ be the cost of making $C$ a cluster, and $B-C$ singletons. This cost is equal to half the number of edges with exactly one endpoint in $B$, plus the number of edges with one endpoint in $B-C$ and the other endpoint in $B$, plus the number of non-edges in $C$. In fact, $cost(B|C)$ is $cost^*(B|C)$ plus half the number of edges that leave $B$. 
\begin{theorem}\label{thm:cost-estimate}
    Let $t^*$ be the threshold among $1,(1+\eps),\ldots, (1+\eps)^{\ceil{\log n}}$ where $cost(B_v|C_{t^*})$ is minimized. If \cref{alg:break-cluster} returns $C_{\tilde{t}}$, then $cost(B_v|C_{\tilde{t}})\le(1+ 219\eps)cost(B_v|C_{t^*})$.
\end{theorem}
\begin{proof}
First note that since $cost(B_v|C_t)$ is $cost^*(B_v|C_t)$ plus half the number of edges that leave $B_v$, for any $t,t'$ we have $cost(B_v|C_t)-cost(B_v|C_{t'}) = cost^*(B_v|C_t)-cost^*(B_v|C_{t'})$, and thus we can use $cost^*$ for comparing the costs, i.e. $t^*$ minimizes $cost^*(B_v|C_t)$ as well as $cost(B_v|C_t)$.

    We prove the Theorem by induction: Suppose that for any $j$, $t_j$ is the threshold among $1,(1+\eps),\ldots, (1+\eps)^{j}$ such that $cost^*(B_v|C_{t_j})$ is minimized, and the output of the for loop in \cref{alg:break-cluster} for $t\in \{1,(1+\eps),\ldots, (1+\eps)^{j}\}$ is $C_{\tilde{t}_j}$. Fix some $i$. Suppose that $cost^*(B_v|C_{\tilde{t}_i})\le(1+ 219 \eps)cost^*(B_v|C_{t_i})$. Note that $C_{\tilde{t}_{i+1}} = \textsc{cost-comparison}(B_v,C_{(1+\eps)^{i+1}},C_{\tilde{t}_i})$. 
    We show that $cost^*(B_v|C_{\tilde{t}_{i+1}})\le(1+ 219 \eps)cost^*(B_v|C_{t_{i+1}})$.

    For ease of notation let $C_1 = C_{t_{i}}$, $C_2 = C_{t_{i+1}}$, and $C' = C_{(1+\eps)^{i+1}}$. So we have $C_2 = \arg \min_{C\in \{C_1,C'\}} cost^*(B_v|C)$. Let $\tilde{C}_1 = C_{\tilde{t}_{i}}$ and $\tilde{C}_2 = C_{\tilde{t}_{i+1}}$. Assuming that $cost^*(B_v|\tilde{C}_1)\le(1+219\eps)cost^*(B_v|C_1)$, we need to show that $cost^*(B_v|\tilde{C}_2)\le(1+219\eps)cost^*(B_v|C_2)$, where $\tilde{C}_2 = \textsc{cost-comparison}(B_v,C',\tilde{C}_1)$.

    Note that $C'\subseteq \tilde{C_1}$. This is because $\tilde{C}_1 = C_{\tilde{t}_i}$ and $C' = C_{(1+\eps)^{i+1}}$ where $\tilde{t}_i\le (1+\eps)^i<(1+\eps)^{i+1}$. We prove the rest of the theorem in the following cases.

    \paragraph{Case 1: $cost^*(B|C')<|C'|/4$.} In this case, by \cref{lem:low-cost} we know that $cost^*(B|C')<cost^*(B|\tilde{C}_1)$ and $\tcost(B|C')<\tcost(B|\tilde{C}_1)$. So $\tilde{C}_2 = C'$. By induction hypothesis, we have $cost^*(B|C')<cost^*(B|\tilde{C}_1)<(1+219\eps) cost^*(B|C_1)$. Since $cost^*(B|C')<(1+219\eps) cost^*(B|C')$ and $C_2\in \{C',C_1\}$, we have $cost^*(B|\tilde{C}_2)=cost^*(B|C')\le (1+219\eps)cost^*(B|C_2)$.
    
    \paragraph{Case 2: $cost^*(B|\tilde{C}_1)<|\tilde{C}_1|/4$:} Similar to case 1, in this case by \cref{lem:low-cost} we know that $cost^*(B|\tilde{C}_1)<cost^*(B|C')$ and $\tcost(B|\tilde{C}_1)<\tcost(B|C')$. So $\tilde{C}_2 = \tilde{C}_1$. By induction hypothesis, we have $cost^*(B|\tilde{C}_1)<(1+219\eps) cost^*(B|C_1)$. Since $cost^*(B|\tilde{C}_1)<(1+219\eps) cost^*(B|C')$ and $C_2\in \{C',C_1\}$, we have $cost^*(B|\tilde{C}_2)=cost^*(B|\tilde{C}_1)\le (1+219\eps)cost^*(B|C_2)$.

    \paragraph{Case 3: $cost^*(B|C')\ge |C'|/4$ and $cost^*(B|\tilde{C}_1>|\tilde{C}_1|/4$:} In this case, by \cref{lem:big-cost} we have that $\tcost(B|C')$ and $\tcost(B|\tilde{C}_1)$ are $(1+219\eps)$ approximations of $cost^*(B|C')$ and $cost^*(B|\tilde{C}_1)$. 

    If $\tcost(B|C')<\tcost(B|\tilde{C}_1)$, then we have $\tilde{C}_2 = C'$. If $C_2 = C'$, then we are done. So assume that $C_2 = C_1$, which means that $cost^*(B|C_1)<cost^*(B|C')$. So we have $cost^*(B|C')\le \tcost(B|C')\le \tcost(B|\tilde{C}_1)\le (1+219\eps)cost^*(B|C_1)$ where the last inequality comes from the induction hypothesis. So we have $cost^*(B|\tilde{C}_2)\le (1+219\eps)cost^*(B|C_2)$

    Similarly if $\tcost(B|C')>\tcost(B|\tilde{C}_1)$, then we have $\tilde{C}_2 = \tilde{C}_1$. If $C_2 = C_1$, then we are done by induction hypothsis. So assume that $C_2 = C'$, which means that $cost^*(B|C_1)>cost^*(B|C')$. So we have $cost^*(B|\tilde{C}_1)\le \tcost(B|\tilde{C}_1)< \tcost(B|C')\le  (1+219\eps)cost^*(B|C')$. So $cost^*(B|\tilde{C}_2)\le (1+219\eps)cost^*(B|C_2)$.

\end{proof}
\begin{theorem}\label{thm:break-cluster-running-time}
    \cref{alg:break-cluster} runs in $|B_v|\poly(\log n,1/\eps)$ time.
\end{theorem}
\begin{proof}
    Note that \cref{alg:break-cluster} calls \cref{alg:cost-comparison} $O(\log n)$ times. Each run of \cref{alg:cost-comparison} takes $|B_v|\poly(\log(n),1/\eps)$ time: this is because estimating $\tcost(B|C_t)$ by \cref{alg:cost-estimate} for a set $C$ takes $O(|C|\log(n)/\eps^3)$ time.
\end{proof}

\subsection{Dynamic version: \textsc{update-cluster($B_v,u$)}}
\label{sec:dynamic-cost-estimate}
The previous section describes how to estimate the cost of clustering $B$ in which $C \subseteq B$ is a cluster, while $B-C$ are singletons.
When a node is inserted, we cannot afford to estimate these costs from scratch. 
Rather, we want to \emph{update} the estimate based on the inserted node.

\subsubsection{Inserting a node into $B - C$}
\label{sec:dynamic-node-to-B-C}
Assume that we already have an estimate of the cost of clustering $B$ as guaranteed by \cref{lemma:cost-estimate}; let $\tcost$ be that estimate.
Assume that a node $z$ is inserted in $B-C$. 
Hence, $z$ will be a singleton.
When $\tcost$ was computed no edge incident to $z$ was in the graph. 
Thus, updating the cost estimate is easy in this case: the cost estimate of clustering $C$ together and $B-C+z$ as singletons equals $\tcost + d(z)$.

\subsubsection{Inserting a node into $C$}
\label{sec:dynamic-node-to-C}
Assume that we already have an estimate of the cost of clustering $B$ as guaranteed by \cref{lemma:cost-estimate}; let $\tcost$ be that estimate.
Assume that a node $z$ is inserted in $C$.
Let $\tcost'$ be the new cost estimate we aim to obtain.
Based on Line~2 of \cref{alg:cost-estimate}, the difference $\tcost' - \tcost$ contains three components: $d(z)$; $-2 \binom{|C| + 1}{2} + 2\binom{|C|}{2}$; and, $3 \cdot \withinClusterEstimate(C+z) - 3 \cdot \withinClusterEstimate(C)$.
The first two components can be updated in $O(1)$ time. 
Even $3 \cdot \withinClusterEstimate(C)$ can be updated in $O(1)$ time -- together with $\tcost$, we store the value of $\withinClusterEstimate(C)$ that led to $\tcost$ itself.
So, it remains to describe how to obtain $\withinClusterEstimate(C+z)$ in $\poly(\log n, 1/\eps)$ time, as we do next.

Let $S_C$ be the value of $S$ at the end of $\withinClusterEstimate(C)$ invocation. 
Let $S_{C+z}$ be a value of $S$ corresponding to $\withinClusterEstimate(C+z)$ that we aim to obtain.
We initialize $S_{C+z} = S_C$, and then update $S_{C+z}$ as follow.

First, given definition of $\tau_C$ on Line~2 of \cref{alg:within-cluster-cost-estimate}, to compute $S_{C+z}$ we sample $5 \cdot \log(n) / \eps^3$ pairs $(v, w) \in (C+z) \times (C+z)$, and for each non-edge $\{v, w\}$ we increment $S_{C+z}$ -- the same as \cref{alg:within-cluster-cost-estimate} does.

However, the $(v, w)$ pairs sampled by \cref{alg:within-cluster-cost-estimate} to compute $S_C$ are simple from $C \times C$, while for $S_{C+z}$ we would like each pair to be sampled from $(C+z) \times (C+z)$. So, second, to account for that, we resample some of the pairs used in computing $S_C$.
This also implies that while our algorithm estimates the cost of $C$ by \cref{alg:within-cluster-cost-estimate}, it also stores in an array all the $\{v, w\}$ pairs it sampled within the for-loop.
Let that array of samples be called $A_C$.

When computing $S_C$, a pair $\{v, w\}$ is sampled with probability $p_C = 1 / \binom{|C|}{2}$. However, when computing $S_{C+z}$, that same pair is sampled with probability $p_{C + 1} = 1/\binom{|C| + 1}{2}$. Let $q = p_C / p_{C+1}$
So, in $A_C$, we resample a pair $\{v, w\}$ with probability $1 - q$, and otherwise, with probability $q$, $\{v, w\}$ is not resampled.
With this process, we have that $\{v, w\} \in C \times C$ remains unchanged with probability $1 / \binom{|C| + 1}{2}$, as desired.

If a pair is resampled, then the new pair is $\{z, u\}$, where $u$ is a node from $C$ chosen uniformly at random.
So, we have
\[
    \prob{\text{$\{z, u\}$ is sampled}}
    = \frac{1}{|C|} \cdot \rb{1 - \frac{\binom{|C|}{2}}{\binom{|C| + 1}{2}}}
    = \frac{1}{|C|} \cdot \frac{2}{|C| + 1}
    = \frac{1}{\binom{|C| + 1}{2}},
\]
as we aim to achieve.
But how many pairs are resampled? How does one choose which pairs to resample in $\poly(\log n, 1/\eps)$ time?

The expected number of resampled pairs is
\[
    \EE{|A_C| \cdot (1 - q)} = \tau_C \cdot \frac{2}{|C|+1} = O(\log(n) / \eps^3).
\]
Hence, by a direct application of the Chernoff bound, with high probability, the number of resampled pairs from $A_C$ is $\poly(\log n, 1/\eps)$.

The final piece elaborates on efficiently finding pairs to resample.
Consider a process that iterates over the elements in $A_C$ and resamples each with probability $1-q$. A downside is that this approach takes $\Theta(|C|)$ time, which is too slow for our goal.
Instead, we observe that the index of \emph{the first} element resampled in $A_C$ is drawn from the geometric distribution with parameter $1-q$.
This observation leads to the following efficient procedure for resampling elements from $A_C$:
\begin{enumerate}
    \item Initialize $i = 0$.
    \item Repeat while $i \le |A_C|$:
    \begin{itemize}
        \item sample an index $j$ from the Geometric distribution with parameter $1 - q$;
        \item $i = i + j$;
        \item if $i \le |A_C|$, resample the $i$-th $\{v, w\}$ pair in $A_C$.
    \end{itemize}
\end{enumerate}
This latter approach enables us to spend time proportional to the number of resampled pairs -- as opposed to $|A_C|$ -- which we know is $\poly(\log n, 1/\eps)$ with high probability.

\begin{lemma}[Dynamic cost estimate of single-cluster + singletons]
\label{lemma:dynamic-cost-estimate}
    There exists an algorithm that, on a node insertion, updates $\costEstimate(B, C)$ in $\poly(\log n, 1/\eps)$ time with high probability.
    The approximation guarantees are the same as those stated in \cref{lemma:cost-estimate}.
\end{lemma}

\subsubsection{\textsc{update-cluster($B_v,u$)}}\label{sec:update-cluster}
Subroutine \textsc{update-cluster($B_v,u$)} is essentially \textsc{Break-cluster} with the costs of $C$ and $B - C$ updated dynamically as explained in \cref{sec:dynamic-node-to-C,sec:dynamic-node-to-B-C}.
%a modification in computing $\tcost$. 

More precisely, when \textsc{Break-cluster} is invoked, we store all the estimates $\tcost(B_v|C_t)$, the sets $B_v$ and $C_t$, for $t=1,(1+\epsilon),\ldots,(1+\epsilon)^{\log n}$.
%: when $v$ is inserted, we run \cref{alg:break-cluster} in $O(d(v))$ time, and save all these $\tcost$s. 
Then, we use \cref{lemma:dynamic-cost-estimate} to update these costs when $u$ is inserted into $B_v$ in $\poly(\log n,\eps)$ time. 
Note that if $d(u)\ge t$, then $u$ joins $C_t$, and otherwise $u$ joins $B_v\setminus C_t$. So each run of \cref{alg:cost-comparison} takes $\poly(\log n, 1/\eps)$, and so \textsc{update-cost}($B_v, u$) works in $\poly(\log n, 1/\eps)$ time.
\section{Recompute}
\label{sec:recompute}
Let $n_0$ be the number of nodes in the graph just after the last recompute. After $\eps n_0 / 6$ updates, we perform a recomputation.
Our \recompute procedure is as follows:
\begin{itemize}
    \item Purge from the database soft-deleted nodes.
    \item Assign to each node $u$ a rank $\pi(u)$ chosen uniformly at random from $[0, 1]$.
    \item Initialize $p(u) = u$ for all nodes $u$.
    \item Sort the nodes in the increasing order with respect to $\pi$.
    \item Insert the nodes, one by one, in this sorted order. The insertions are handled by \cref{alg:node-dynamic-pivot}, except that a new $\pi(u)$ value is not obtained within \cref{alg:node-dynamic-pivot}, but is used the one computed in the first step of \recompute.
\end{itemize}
The nodes are processed in the ordering based on their $\pi$ values for the following reasons. 
Given a node $u$, \cref{alg:node-dynamic-pivot} performs updates or exploration only when ranks are smaller than $\pi(u)$. In particular, the node $v$ defined in that algorithm is used only if $\pi(v) \le \pi(u)$.

The running time of \recompute is a constant factor of the running time used to process the insertions.
To see that, charge a recomputation running time to the $\eps n_0 / 6$ most recent updates. 
Observe that this kind of charge is applied to each update during only one \recompute.
Hence, at most $n_0 + \eps n_0 / 6 < 2n_0$ insertions are charged to $\eps n_0 / 6$ updates. So, each update is charged $12 / \eps$ insertions. 
Since an insertion takes $\poly(\log n, 1/\eps)$ amortized time, this additional charge also takes $\poly(\log n, 1/\eps)$ time per update.

\section{Auxilary Lemmas}
\begin{lemma}\label{lem:poornodes}
    Let $\alpha<1$ be a constant. 
    Let $A$ be a clustering algorithm, where the probability of a node $v$ being a pivot over all orderings is at most $1/d(v)$. 
    The expected sum of the degrees of $(A,\alpha)$-poor nodes is at most $4\alpha\frac{1+3/2 \cdot \alpha}{1-5/2 \cdot \alpha}$ times the total expected cost of $A$. 
\end{lemma}
\begin{proof}%[Proof of \cref{lem:poornodes}]
 %\mina{replace all $C_v$  with $C_{ref}^\pi(v)$, and $ref$ with $A$.}
Let $N_{low}(v)$ be the neighbors of $v$ with degree at most $\alpha d(v)$. 
Observe that $N_{low}(v)$ is independent of $\pi$ and $A$.
%In particular, if $v$ is pivot in $A$, $u \in C_v \cap N_{poor}(v)$, then $u$ is a poor node by definition.

\paragraph{The sum of poor-node degrees in a cluster.}
First, we show that for any $\pi$ and $v$ such that $v$ is a pivot in $A$ wrt to $\pi$, the sum of the degrees of the $(A, \alpha)$-poor nodes in $v$'s cluster is at most $\alpha d(v) \cdot \min(3\alpha d(v),|N_{low}(v)|)$.

Fix a ordering $\pi$ of nodes where $v$ is a pivot in $A$. 
Note that if $N_{low}(v)=\emptyset$, then there are no $(A, \alpha)$-poor nodes in $v$'s cluster, and the sum of degrees of all its $(A, \alpha)$-poor nodes is zero. 
Hence, assume $N_{low}(v) \neq \emptyset$.
The size of $v$'s cluster is at most $3\alpha d(v)$ since the cluster size is at most three times the degree of an $(A, \alpha)$-poor node. 
The latter is the case as an $(A, \alpha)$-poor node $u$ in a cluster $C$ is \textbf{not} light by definition. Hence, $d_C(u) > |C| / 3$, which further implies $|C| < 3 d_C(u) \le 3 d(u) \le 3 \alpha d(v)$, where $d(u) \le \alpha d(v)$ be definition of $(A, \alpha)$-poor nodes.
So, the sum of degrees of all $(A, \alpha)$-poor nodes in $v$'s cluster is at most $\alpha d(v) \cdot \min(3\alpha d(v), |N_{low}(v)|)$.

\paragraph{The expected sum of poor-node degrees in $A$.}
Let $Q^\pi$ be the set of $(A, \alpha)$-poor nodes wrt the ordering $\pi$. Let $Pivot_v$ be the event that $v$ is a pivot in $A$.
%\snote{Maybe shorten this by directly writing $\sum_v \alpha d(v) \cdot \min(3\alpha d(v), |N_{poor}(v)|) \cdot \prob{Pivot_v}$.}
We have %\mina{replace $cost_{sing}$ with degree}
% \begin{align*}
%     \EE[\pi]{\sum_{u\in Q^\pi}d(u)} & = \sum_{\pi} \frac{1}{n!} \sum_{v\in P_{ref}^\pi} \sum_{u\in Q^\pi \cap C_v^\pi}d(u) \\
%     & = \sum_{v}\sum_{\pi\ :\ v\in P_{ref}^\pi} \frac{1}{n!} \sum_{u\in Q^\pi \cap C_v^\pi}d(u) \\
%     & \le \sum_{v: N_{poor}(v)\neq \emptyset} \mathbb{P}(Pivot_v)\alpha d(v) \min(3\alpha d(v),|N_{poor}(v)|).
% \end{align*}
\begin{align*}
    \EE[\pi]{\sum_{u\in Q^\pi}d(u)} & = \sum_{\pi} \frac{1}{n!} \sum_{v\in P_{ref}^\pi} \sum_{u\in Q^\pi \cap C_v^\pi}d(u) \\
    & = \sum_{v}\sum_{\pi\ :\ v\in P_{ref}^\pi} \frac{1}{n!} \sum_{u\in Q^\pi \cap C_v^\pi}d(u) \\
    & \le \sum_{v}\sum_{\pi\ :\ v\in P_{ref}^\pi} \frac{1}{n!} \cdot \alpha d(v) \cdot \min(3\alpha d(v),|N_{low}(v)|) \\
    & = \sum_{v}\rb{\alpha d(v) \cdot \min(3\alpha d(v),|N_{low}(v)|) \cdot \sum_{\pi\ :\ v\in P_{ref}^\pi} \frac{1}{n!}} \\
    & = \sum_{v} \prob{Pivot_v} \cdot \alpha d(v) \cdot \min(3\alpha d(v),|N_{low}(v)|).
\end{align*}

%Ideally, we would like to claim that $\mathbb{P}(Pivot_v)\le \frac{1}{d(v)}$. 
%While this is the case for the reference clustering, this is not the case for our algorithm. To see that, consider the following case. First, $k$ nodes are inserted, each without neighbors. Therefore, the algorithm makes all those $k$ nodes pivot. Second, a single node $s$ is inserted, and $s$ is connected to all the $k$ nodes. In expectation, $k/2$ of those nodes should stop being pivots. However, our algorithm does not have the running time budget to update the pivots properly, and hence $k/2$ nodes remain being pivots while they are not in the reference clustering.
%\mina{Needs some more argument about fake pivots, basically need to condition that v is not a fake pivot}. \snote{I think we can also avoid conditioning and split the sum into $P_{ref} \cap P_{alg}$ and $P_{alg} \setminus P_{ref}$. The probability thing holds for $P_{ref} \cap P_{alg}$. So, the question is how we handle fake pivot.} 
Note that $\prob{Pivot_v} \le \frac{1}{d(v)}$ in $A$.
Hence, the expected sum of degrees of $(A, \alpha)$-poor nodes in $A$ is at most 
\begin{equation}\label{eq:sum-poor-nodes}
    \sum_{v} \alpha \cdot \min(|N_{low}(v)|, 3\alpha d(v)) = \sum_{v, N_{low}(v) \neq \emptyset} \alpha \cdot \min(|N_{low}(v)|, 3\alpha d(v)).
\end{equation}
% And hence by \cref{lem:singleton} the total cost of type 2 light nodes in our algorithm is at most 
% $$
% (1+\epsilon)\sum_{v, N_{l2}(v)\neq \emptyset} 2\alpha \min(|N_{l2}(v)|, \alpha d(v)).
% $$

\paragraph{Lower bounding the cost of $A$.}
Now, we compare the above value to the cost of $A$. (The analysis we provide applies to the cost of any clustering, even the one incurred by the optimal solution.)
Fix a ordering $\pi$. All the costs below are defined wrt $\pi$, and we avoid the superscript $\pi$. 

Recall that for a vertex $w$, $C_A(w)$ denotes the cluster of $w$ in $A$. 
We define a cost function $\widehat{cost}_{A}$ to redistribute the cost $cost_{A}$ as follows
\[
    \widehat{cost}_{A}(w) \eqdef \frac{1}{2}\rb{cost_{A}(w)+\frac{\sum_{u\in C_A(w)} cost_{A}(u)}{|C_A(w)|}}.
\]
In other words, in $\widehat{cost}_{A}$, a node $w$ distributes $1/2 \cdot cost_{A}(w)$ over the $|C_A(w)|$ nodes in $C_A(w)$, and $w$ itself pays for the remaining $1/2 \cdot cost_{A}(w)$.
Note that $2 \cdot \widehat{cost}_{A}(w)\ge {cost}_{A}(w)$ and $\sum_{w} cost_{A}(w) = \sum_{w} \widehat{cost}_{A}(w)$. 

Now we compute $\EE{\widehat{cost}_{A}(v)}$ for $v$ with $N_{poor}(v)\neq \emptyset$. 
Note that we are not conditioning on $v$ being a pivot here. Let $t = |N_{poor}(v)|$, and let $t' = \min(t, 3\alpha d(v))$. 
We consider three cases based on $|C_A(v)|$:
\begin{itemize}
    \item Case $|C_A(v)|\ge d(v)+\frac{t'}{2}$: then, there are at least $\frac{t'}{2}$ non-neighbors of $v$ in $C_A(v)$ and hence $cost_{A}(w)\ge \frac{t'}{2}$. Also, $2 \cdot \widehat{cost}_{A}(w)\ge{cost}_{A}(w)\ge \frac{t'}{2}$.
    \item Case $|C_A(v)|\le d(v)-\frac{t'}{2}$: then, at least $\frac{t'}{2}$ neighbors of $v$ are outside of $C_A(v)$ and hence $cost_{A}(w)\ge \frac{t'}{2}$. Again, $2 \cdot \widehat{cost}_{A}(w)\ge{cost}_{A}(w)\ge \frac{t'}{2}$.
    \item Case $d(v)-\frac{t'}{2}\le|C_A(v)|\le d(v)+\frac{t'}{2}$: 
    Let $t''=|N_{low}(v)\cap C_A(v)|$. The number of non-neighbors of a node $u\in N_{poor}(v)\cap C_A(v)$ is least $|C_A(v)|-d(u)\ge d(v)-t'/2-\alpha d(v)\ge (1-5/2 \cdot \alpha)d(v)$, where we used the fact that $t'/2 \le 3/2 \cdot \alpha d(v)$. 
    So
    \begin{align*}
        \sum_{u\in C_A(v)} cost_{A}(u) & \ge \sum_{u\in N_{low}(v)\cap C_A(v)} cost_{A}(u) \ge t''(1-5/2 \cdot \alpha)d(v). 
    \end{align*}
    Moreover, for each node $u\in N_{low}(v)\setminus C_A(v)$, $v$ incurs one unit of cost. Also, by definition, $t \ge |N_{low}(v)|$. 
    Therefore, $cost_{A}(v)\ge t-t''$. 
    Recall that this case assumes $|C_A(v)| \le d(v) + t'/2 \le d(v) + 3/2 \cdot \alpha d(v) = (1 + 3/2 \cdot \alpha) d(v)$.
    This yields
    \begin{align}
        2 \cdot \widehat{cost}_{A}(v) & \ge t-t''+\frac{t''(1-5/2 \cdot \alpha)d(v)}{|C_A(v)|} \ge t-t'' + t''\frac{1-5/2 \cdot\alpha}{1+3/2 \cdot\alpha}. \label{eq:2-hcost-A}
    \end{align}
    Observe that for $x \in [0, 1]$ and $t \ge t''$ we have $(1 - x) t \ge (1 - x) t''$, and hence $t - t'' + t''x \ge t x$. Thus, \cref{eq:2-hcost-A} implies
    \[
        2 \cdot \widehat{cost}_{A}(v) \ge t\frac{1-5/2 \cdot\alpha}{1+3/2 \cdot\alpha} \ge t'\frac{1-5/2 \cdot\alpha}{1+3/2 \cdot\alpha}.
    \]
    \end{itemize}
Therefore, in any case, $\widehat{cost}_{A}(v)\ge \frac{1}{4}t'\frac{1-5/2 \cdot\alpha}{1+3/2 \cdot\alpha}$. 
So, the total cost of clustering $A$ is at least
$$
\sum_{v, N_{poor}(v)\neq \emptyset} \frac{1}{4}\min(|N_{poor}(v)|, 3\alpha d(v))\frac{1-5/2 \cdot\alpha}{1+3/2 \cdot\alpha}.
$$
Note that the above lower bound is deterministic; in particular, it holds for any randomness used by $A$.
Comparing this lower bound to \cref{eq:sum-poor-nodes}, we conclude that the expected sum of degrees of poor nodes is at most $4\alpha\frac{1+3/2 \cdot \alpha}{1-5/2 \cdot \alpha}$ times the (expected) total cost of $A$, as advertised by the claim.
\end{proof}

\begin{lemma}[Cost of making light and heavy singletons]\label{lem:heavynode}
    Consider a fixed ordering $\pi$ and a clustering algorithm $A$. 
    Let $A'$ be the algorithm that first runs $A$ and, after, makes a subset of $A$-light and $A$-heavy nodes singletons. 
    The cost of $A'$ is at most $\frac{\beta+1}{\beta-1}$ the cost of $A$.
    %if $u$ is a heavy node in the reference clustering, then the cost of $u$ in the algorithm is at most $\frac{\beta(1+\epsilon)}{\beta-1}$ times the cost of $u$ in the reference clustering.
\end{lemma}
\begin{proof}
    For any algorithm $B$, let $d_{in}^B(u)$, $d_{out}^B(u)$ and $\bar{d}_{in}^{B}(u)$ be the number of neighbors of $u$ inside its cluster, number of neighbors of $u$ outside its cluster, and the number of non-neighbors of $u$ inside its cluster, respectively. 
    Note that since $\pi$ is fixed, algorithm $B$ determines the clusters. Let $cost_B$ denote the cost of algorithm $B$ (with respect to $\pi$). 
    We have $2 \cdot cost_B = \sum_u [d_{out}^B(u)+\bar{d}_{in}(u)]$. Let $L$ be the set of $A$-light nodes that $A'$ makes singletons, and $H$ be the set of $A$-heavy nodes that $A'$ makes singletons.

    We can write the cost of algorithm $A'$ as follows:
    \begin{align*}
    2 cost_{A'} &= \sum_{u\in L\cup H} d(u)+\sum_{u\notin L\cup H} [d_{out}^{A'}(u)+\bar{d}_{in}^{A'}(u)] \le \sum_{u\in L\cup H} [d(u)+d_{in}^A(u)]+\sum_{u\notin L\cup H} [d_{out}^{A}(u)+\bar{d}_{in}^{A}(u)] 
    \end{align*}
    where the inequality comes from the fact that going from $A'$ to $A$, any new cost on a node $u\notin L\cup H$ is due to an edge inside $u$'s cluster in $A$ that it attached to a node in $L\cup H$.
    We prove that $\sum_{u\in L\cup H} [d(u)+d_{in}^A(u)] \le \frac{\beta+1}{\beta-1}\sum_{u\in L\cup H} [d_{out}^{A}(u)+\bar{d}^{A}_{in}(u)]$. Given this, we will have that $cost_{A'}\le \frac{\beta+1}{\beta-1}cost_A$.
    
    If $u$ is $A$-heavy, then $d(u)\ge \beta|C_A(u)|\ge \beta d_{in}^A(u)$. So $d_{out}^A(u)\ge (\beta-1)d_{in}^A(u)$, and hence $d(u)+d_{in}^A(u)=d^A_{out}(u)+2d_{in}^A(u)\le \frac{\beta+1}{\beta-1}d_{out}^A(u)$.
    
    If $u$ is a light node, then we have that $d_{in}^{A}(u)\le |C_A(u)|/3$, so $d_{in}^{A}(u)\le \frac{1}{2}\bar{d}_{in}^A(u)$, and hence $d(u)+d^A_{in}(u)=d^A_{out}(u)+2d^A_{in}(u)\le d^A_{out}(u)+\bar{d}^A_{in}(u)$. 
    % let $u$ be a heavy node. The cost of $u$ in the reference clustering 
    % This comes from the fact that if we make $u$ a singleton in the reference clustering, the cost is multiplied by at most $\frac{\beta}{\beta-1}$. Then by Lemma \ref{lem:singleton} the cost of our algorithm is at most $(1+\epsilon)$ the cost of making $u$ a singleton.
\end{proof}

\begin{lemma}\label{lem:bad-cluster-cost}
    Consider a fixed ordering $\pi$ and a clustering algorithm $A$, and consider a $A$-bad cluster $C_{A}^\pi$. 
    Then, its cost is at least $\frac{2}{3}(1-\gamma)|C_{A}^\pi|^2$.
\end{lemma}
\begin{proof}
    Recall that a bad cluster does not have any $(ref,\alpha)$-poor nodes. 
    So, all the bad nodes in a bad cluster are either light or heavy. 
    Let $C=C_{ref}^\pi$. 
    
    For a light node $u\in C$, the cost of $u$ is at least $2|C|/3$, since $d_C(u)\le |C|/3$. 
    For a heavy node $u\in C$, the cost of $u$ is at least $(\beta-1)|C|/2$, since $u$ has at least $(\beta-1)|C|$ neighbors outside $C$. 
    So the cost of any bad node in $C$ is at least $\min{\{(\beta-1)/2,2/3\}} \cdot |C|$. 
    There are at least $(1-\gamma)|C|$ many bad nodes in $C$. 
    Thus, the total cost of $C$ is at least $\min{\{(\beta-1)/2,2/3\}} \cdot (1-\gamma) \cdot |C|^2$. Since $\beta\ge4$, $\min{\{(\beta-1)/2,2/3\}}=2/3$. 
\end{proof}

\section{Ommited proofs}
\subsection{Proof of \cref{lem:pivot-set}}
%\begin{proof}[Proof of \cref{lem:pivot-set}]
    %\stodo{$L$ vs $L'$, check this}
    % We start by proving the first part of the claim.
    % \paragraph{If $v\in P_{ref}^\pi$, then $\pi(v)\le L'/d(v)$ holds with probability at least $1-d(v) \cdot e^{-L'}$.}
   
    It holds that $\probgiven{\pi(v)\le L'/d(v)}{v\in P_{ref}^\pi} = 1 - \probgiven{\pi(v) > L'/d(v)}{v\in P_{ref}^\pi}$, and we upper-bound the latter probability:
    \begin{align*}
        \probgiven{\pi(v) > L'/d(v)}{v\in P_{ref}^\pi}
        & = \frac{\prob{\pi(v) > L'/d(v) \text{ and } v\in P_{ref}^\pi}}{\prob{v\in P_{ref}^\pi}} \\
        & = d(v) \cdot \prob{\pi(v) > L'/d(v) \text{ and } \forall u \in N(v): \pi(u) > \pi(v)} \\
        & \le d(v) \cdot \rb{1 - \frac{L'}{d(v)}}^{d(v) + 1} \\
        & \le d(v) \cdot e^{-L'}.
    \end{align*}
    In the derivation, we used the fact that the entries of $\pi$ are chosen independently and uniformly from range $[0, 1]$. Now note that $d(v)e^{-L'}\le ne^{c\log n}<n^{-c+1}$, so 
    Hence, $\probgiven{\pi(v)\le L'/d(v)}{v\in P_{ref}^\pi} \ge 1 - n^{-(c-1)}$, as desired.
%
    % \paragraph{$P_{ref}^\pi \subseteq P_{alg}^\pi$ with high probability.}
    % Hence, for $L = (c + 2) \ln n$, and since $d(v) \le n$, we have that $\probgiven{\pi(v)\le L/d(v)}{v\in P_{ref}^\pi} \ge 1 - n^{-c-1}$.
    % Moreover, by the union bound, for all $v \in P_{ref}^\pi$ simultaneously, we have that $\pi(v) \le L / d(v)$ with probability at least $1-n^{-c}$.
%\end{proof}

\subsection{Proof of \cref{lem:1-eps-good-nodes}}

%\begin{proof}[Proof of \cref{lem:1-eps-good-nodes}]
%\stodo{$L$ vs $L'$, check this}
%\stodo{I suggest we move this proof to the appendix. So, no need to fix the formula display for now.}
    Let $L'=L/2\beta$. Let $A$ be the set of first $(1-\epsilon)C_{ref}(v)[good]$ good nodes in $C_{ref}(v)$ and let $B$ be the set of last $\epsilon C_{ref}(v)[good]$ good nodes in $C_{ref}(v)$. So $C_{ref}(v)[good] = A\cup B$.

First we introduce some notation: For two nodes $u$ and $w$ where $w$ is inserted before $u$, the degree of a vertex $w$ at the time where $u$ is inserted is denoted by $d^{(u)}(w)$. Recall that the current degree of $w$ is denoted by $d(w)$, and $d(w)\ge d^{(u)}(w)$.
    
    Next, note that if the pivot $v$ comes after all the nodes in $A$, then since $v$ is a pivot in \refClustering by \cref{lem:pivot-set} we have that whp $\pi(v)\le L'/d(v)\le L/d(v)$  and so it scans its neighborhood and invokes \textsc{explore$(v)$} which assigns all these nodes to $v$ as their pivot. So suppose that $v$ comes before the nodes in $B$.
    
    Now consider a good node $u \in C_{ref}(v)[good]$ and suppose that $u$ comes after $v$ in the dynamic ordering. Since $u$ is not heavy, we have $d(u)\le \beta|C_{ref}(v)|\le \beta d(v)$.
    % First, since $u$ is not light, we have $d(u)\ge d_{C_v}(u)\ge |C_v|/3\ge |C_v[good]|/3$. Next, since $u$ is not heavy, we have $d(u)\le \beta |C_v|$. $C_v$ is a good cluster, so $|C_v|\le \frac{1}{\rho}|C_v^{good}|$. So we have $|C_v^{good}|/3\le d(u)\le \frac{\beta}{\rho}|C_v^{good}|$.

    We compute the probability that $u$ is assigned to $v$'s cluster and also that $u$ invokes \textsc{explore$(v)$}. %\mina{we also need that at the time of the insertion of $u$, $d(u)\ge \alpha d(v)$.}
    In particular, we want to lower-bound the probability that $\pi(v)<\pi(u)\le \frac{L}{d^{(u)}(u)}$ and $d^{(u)}(v)\le \frac{L}{\pi(u)}$, conditioned on $\pi(v)<\pi(u)$, i.e., in \refClustering $u$ is in $C_v$. 
    
    First note that if $d(v)\le L$ and $d(u)\le L$, then both these conditions are satisfied. So we assume that $\max(d(u),d(v))>L>L'\beta$.

    Next, recall that since $v$ is a pivot in the \refClustering, by \cref{lem:pivot-set}, we have that $\pi(v)\le L'/d(v)$ holds with probability $1-1/n^{c-1}$ for a large constant $c$ by \cref{lem:pivot-set}. Moreover, since $d^{(u)}(v)\le d(v)$, the probability that $d^{(u)}(v)\le \frac{L}{\pi(u)}$ is at least the probability that $d(v)\le \frac{L}{\pi(u)}$. Similarly, the probability that $\pi(u)\le \frac{L}{d^{(u)}(u)}$ is at least the probability that $\pi(u)\le \frac{L}{d(u)}$ as $d^{(u)}(u)\le d(u)$.
    So:
    \begin{align*}
        & \prob{u \text{ invokes \textsc{explore($v$)}}\ |\ u\in C_{ref}(v)} \\ 
        = & \prob{\pi(v)<\pi(u)\le \frac{L}{d^{(u)}(u)} \text{ and } d^{(u)}(v)\le \frac{L}{\pi(u)}\ |\ \pi(v)<\pi(u)\le 1} \\
        \ge & \prob{\pi(v)<\pi(u)\le \frac{L}{d(u)} \text{ and } d(v)\le \frac{L}{\pi(u)}\ |\ \pi(v)<\pi(u)\le 1} \\
        = & \prob{\pi(v)<\pi(u)\le \frac{L}{\max(d(u),d(v))}\ |\ \pi(v)<\pi(u)\le 1} \\
        = & \prob{\pi(v)<\pi(u)\le \frac{L}{\max(d(u),d(v))}\ |\ \pi(v)<\pi(u)\le 1, \pi(v)\le L'/d(v)} \cdot \prob{\pi(v)\le L'/d(v)} \\
        & + \prob{\pi(v)<\pi(u)\le \frac{L}{\max(d(u),d(v))}\ |\ \pi(v)<\pi(u)\le 1, \pi(v) > L'/d(v)} \cdot \prob{\pi(v) > L'/d(v)} \\
        \ge &  \prob{\frac{L'}{d(v)}<\pi(u)\le \frac{L}{\max(d(u),d(v))}\ |\ \frac{L'}{d(v)}<\pi(u)\le 1} \cdot \rb{1 - n^{-c+1}} \\
        % \ge &  \prob{\frac{L'\beta}{d(u)}<\pi(u)\le \frac{L}{d(u)}\ |\ \frac{L'\beta}{d(u)}<\pi(u)\le 1} \cdot \rb{1 - n^{-c+1}} \\
        % = & \frac{L-L'\beta}{d(u)-L'\beta} \cdot \rb{1 - n^{-c+1}}
        %     \ge \frac{L-L'\beta}{d(u)} \cdot \rb{1 - n^{-c+1}} \\
        %     & \ge \frac{L/2}{d(u)} \cdot \rb{1 - n^{-c+1}}
    \end{align*}
    If $d(v)\le d(u)$, then since $d(u)\le \beta|C_{ref}(v)|\le \beta d(v)$ and  $L'\beta = L/2$, we have
\begin{align*}
\prob{u \text{ invokes \textsc{explore($v$)}}\ |\ u\in C_{ref}(v)}
             & \ge \prob{\frac{L'\beta}{d(u)}<\pi(u)\le \frac{L}{d(u)}\ |\ \frac{L'\beta}{d(u)}<\pi(u)\le 1} \cdot \rb{1 - n^{-c+1}} \\
        & = \frac{L-L'\beta}{d(u)-L'\beta} \cdot \rb{1 - n^{-c+1}}
            \ge \frac{L-L'\beta}{d(u)} \cdot \rb{1 - n^{-c+1}} \\
            & \ge \frac{L/2}{d(u)} \cdot \rb{1 - n^{-c+1}}\ge \frac{\alpha L/2}{d(u)} \cdot \rb{1 - n^{-c+1}}
\end{align*}
    Nnote that $d(u)-L'\beta>0$ since we assume that $\max(d(u),d(v))>L>L'\beta$. 
    
    If $d(v)\ge d(u)$, we use the fact that since $u$ is not $(\alpha,A)$-poor, we have $d(u)\ge \alpha d(v)$, and so
\begin{align*}
    \prob{u \text{ invokes \textsc{explore($v$)}}\ |\ u\in C_{ref}(v)}
             & \ge \prob{\frac{L'}{d(v)}<\pi(u)\le \frac{L}{d(v)}\ |\ \frac{L'}{d(v)}<\pi(u)\le 1} \cdot \rb{1 - n^{-c+1}} \\
             & \ge \frac{L-L'}{d(v)-L'}\cdot \rb{1 - n^{-c+1}} \ge \frac{L/2}{d(v)}\cdot \rb{1 - n^{-c+1}}\\
            & \ge  \frac{\alpha L/2}{d(u)}\cdot \rb{1 - n^{-c+1}} 
\end{align*}
    %where $c$ is the ``whp'' constant from \cref{lem:pivot-set}, 
    %and we use the fact that $L'\beta = L/2$. 
    The above inequality holds for each $u\in B$. We show that with high probability, a node in $B$ invokes \textsc{explore($v$)}. First recall that $C_{ref}(v)$ is a good cluster, so $|C_{ref}(v)[good]|\ge \gamma |C|\ge \gamma d(u)/\beta$. Thus we have $|B|\ge \gamma \epsilon d(u) /\beta$. Let $t = \frac{\epsilon \gamma\alpha}{4 \beta}$. Since $(1-n^{-c+1})>1/2$, we have:
\begin{equation*}
    \prob{u \text{ invokes \textsc{explore($v$)}}} \ge \frac{tL}{|B|}
\end{equation*}
So the probability that none of $u\in B$ invokes \textsc{explore($v$)} is at most $(1-\frac{tL}{|B|})^{|B|}\le e^{-tL} \le n^{-c}$, where we use $L\ge \frac{4c\beta}{\epsilon \gamma\alpha}\log n = \frac{c}{t}\log n$. So with probability $1-n^{-c}$, \textsc{explore($v$)} is invoked by a node in $B$, and so all the nodes in $A$ are going to be assigned to $v$ as their pivot.
    %We show that whp one of them has to invoke \textsc{Eplore}. Note that since $C_{ref}(v)$ is a good cluster, we have that $|C_{ref}(v)[good]|\ge \gamma |C|\ge \gamma d(u)/\beta$, and 
    %so $|B|\ge \epsilon \gamma d(u)/\beta$. 
    %\mina{we can let $L'/2 = L\beta = \log n$. So the exploration prob is at least $\log n/d(u)$, and we have $d(u)\epsilon/\beta$ samples. So we proved one of the last $\epsilon |C_{ref}(v)[good]|$ nodes does exploration, since $u$ is not poor, $v$ does exploration, and so all the misplaced nodes coming before $u$ assign $v$ as their pivot.}.
%\end{proof}

\subsection{Proof of \cref{lem:4-approx-break-cluster}}
%\begin{proof}[Proof of \cref{lem:4-approx-break-cluster}]
    Let $k=C^*$ and let $t$ be the smallest power of $(1+\epsilon)$ no smaller than $\frac{2k}{3}$. So $\frac{2k}{3}\le t\le \frac{2k(1+\epsilon)}{3}$. Let $S = C^*\setminus C_{t}$, and let $T=C_{t}\setminus C^*$. We refer to the clustering where $C_t$ is one cluster and all $B_v\setminus C_t$ is singleton as $\mathscr{C}_1$ and the clustering where $C^*$ is one cluster and all $B_v\setminus C^*$ is singleton as $\mathscr{C}^*$. 
    %Note that the cost of the edges going from inside of $B_v$ to outside of $B_v$ is the same in both clusterings, so we ignore them. 

    The cost of $\mathscr{C}_1$ and $\mathscr{C}^*$ differ in edges and non-edges with one endpoint in $S$ or $T$. They share any other cost associated to an edge or non-edge that is in disagreement with the clustering. %Moreover, they also share the edges between $S$ and $T$ in the cost. 
    So we only need to compare the excess cost that is not part of this common cost.
    
    By the cost of a node $u\in S\cup T$, we mean the number of $v\notin S\cup T$ where $uv$ is in disagreement with the clustering plus half the number of $v\in S\cup T$ where $uv$ is in disagreement with the clustering. %The cost of a node $u\in T$ is defined similarly, which is the number of $v\notin T$ where $uv$ is in disagreement with the clustering plus half the the number of $v\in T$ where $uv$ is in disagreement with the clustering . 
    We let $cost(u)_{\mathscr{C}}$ denote the cost of $u$ in a clustering $\mathscr{C}\in \{\mathscr{C}_1,\mathscr{C}^*\}$. Note that the total cost of $\mathscr{C}$, indicated by $cost(\mathscr{C})$, is the sum of the individual costs of $u\in S\cup T$, plus the common cost between $\mathscr{C}_1$ and $\mathscr{C}^*$.

    First, assume that $|T|\le k$. We show that $cost(\mathscr{C}_1)\le \frac{4}{1-2\epsilon}cost(\mathscr{C}^*)$.

    Consider a node $u\in S$. In $\mathscr{C}_1$, the cost of $u$ is at most $t$ since each node in $S$ has degree at most $t$, and $u$ is a singleton in $\mathscr{C}_1$. In $\mathscr{C}^*$, the cost of $u$ is at least $\frac{k-t}{2}$ since there are at least $k-t$ non-edges attached to $u$ in $C^*$. Note that for $t\le \frac{2k(1+\epsilon)}{3}$, we have $cost(u)_{\mathscr{C}_1}\le t\le \frac{4}{1-2\epsilon}\cdot (k-t)/2 \le \frac{4}{1-2\epsilon}cost(u)_{\mathscr{C}^*}$.

    To compute the cost of $u\in T$, for a set $C$ recall that $d_{C}(u)$ is the degree of $u$ inside $C$. So $d(u)-d_{C_{t}}(u)$ is the number of edges attached to $u$ outside $C_{t}$. In $\mathscr{C}_1$ the cost of $u$ is the number of nodes attached to $u$ outside of $C_t$ plus the number of nodes $v$ not attached to $u$ inside $C_t\setminus T$, plus half of the number of nodes $v\in T$ that are not attached to $u$. So
    \begin{align*}
        cost(u)_{\mathscr{C}_1}&= d(u)-d_{C_{t}}(u)+(|C_{t}|-|T|-d_{C_{t}\setminus T}) + (|T|-1-d_{T}(u))/2\\
        & \le  d(u)-d_{C_{t}}(u)+(|C_{t}|-|T|-d_{C_{t}\setminus T}) + (|T|-1-d_{T}(u)) \\
        &= |C_{t}| + d(u)-1-2d_{C_{t}}(u) \\
        & \le |T|+|C^*\cap C_{t}|+d(u)-2d_T(u) \\
        & \le 2k+t-2d_T(u)
    \end{align*}
    Meanwhile we have that there are at least $d(u)-d_{T}(u)$ nodes attached to $u$ outside of $T$, so we have  
    \begin{align*}
        cost(u)_{\mathscr{C}^*} = d(u)-d_{T}(u) + d_T(u)/2 \ge \frac{t-d_T(u)}{2}
    \end{align*}
    Note that when for $t\ge 2k/3$, we have $ 2k+t-2d_T(u) \le 4(\frac{t-d_T(u)}{2})$, so $cost(u)_{\mathscr{C}_1} \le 4cost(u)_{\mathscr{C}^*}$. And so $cost(\mathscr{C}_1)\le \frac{4}{1-2\epsilon}cost(\mathscr{C}^*)$.

    Now suppose that $|T|> k$. Let cluster $\mathscr{C}_2$ be the clustering where all the nodes in $B_v$ are singleton. In fact, $\mathscr{C}_2$ is the clustering for when $t=n$. In this case, the cost of $\mathscr{C}_2$ is the cost of $\mathscr{C}^*$ plus the number of edges between the nodes in $C^*$. We have $cost(\mathscr{C}_2)\le cost(\mathscr{C}^*)+k^2/2$. Now to bound $cost(\mathscr{C}^*)$, the cost of $cost(\mathscr{C}^*)$ is at least the number of edges with at least one endpoint in $T$. This value is at least $t|T|/2$. So $cost(\mathscr{C}^*)\ge t|T|/2$. Now since $|T|\ge k$ and $t\ge 2k/3$, we have $3\cdot t|T|/2\ge k^2/2$, so $cost(\mathscr{C}_2) \le cost(\mathscr{C}^*)+k^2/2 \le cost(\mathscr{C}^*)+3\cdot t|T|/2\le 4cost(\mathscr{C}^*)$.
%\end{proof}

\section{Running Time Analysis}
%\mina{could go in the appendix as we provide insight in the intro}
In \cref{sec:recompute}, we show that our \textsc{recompute} subroutine takes $\poly(\log n, 1/\eps)$ amortized time per update. 
%If the number of nodes in the last run of \textsc{recompute} was $N$, we run \textsc{recompute} when we have $O(\epsilon N)$ updates. To see why this only gives us $O(\poly(\log N)/\eps)$ amortized update time, Note that after $O(\eps N)$ updates the number of nodes are at most $(1+\eps)N$, so the recompute will take $O(N\poly(\log N))$ time. This time divided by the $\eps N)$ steps between the two recomputes yields $O(\poly(\log N)/\eps)$ amortized update time. Since $N<T$ where $T$ is the total number of updates, we have the amortized update time due to recompute is $O(\poly(\log T)/\eps)$. 
%Let $T$ be the total number of updates.
Next, we analyze the running time of \cref{alg:node-dynamic-pivot} and prove \cref{thm:running-time}.

\begin{proof}
    We first introduce some notation: Let the degree of a node $w$ at the time of the insertion of node $u$ be $d^{(u)}(w)$. 
    Recall that $d(w)$ is the current degree of $w$ and $d(w)\ge d^{(u)}(w)$.

    First note that if $\pi(u)>L/d^{(u)}(u)$, then the running time is $O(\log n)$ plus the running time of \textsc{update-cluster}, which is in $\poly(\log n,\frac{1}{\eps})\le \poly(\log T,\frac{1}{\eps})$ time; see \cref{sec:update-cluster} for details.

    Next, consider the case where $\pi(u)\le L/d^{(u)}(u)$. For each such $u$, we set aside a budget of $L/\pi(u)\cdot \poly(\log n, 1/\eps)$.
    In this case, the algorithm scans all the neighbors of $u$, which takes $d^{(u)}(u)<L/\pi(u)$ time. 
    
    Next, we analyze the running time of \textsc{explore}. For this, consider a pivot node $v$. Note that when we run \textsc{explore} on a node $v$, we scan all its neighbors, and if a neighbor $w$ is also a pivot, we scan all of the neighbors of $w$ as well and remove $w$ from being a pivot. 
    Observe that once a node is removed from being a pivot, it can never be a pivot until the next recompute, when nodes get new ranks. 
    We pay for scanning the neighbors of $w$ from the budget of $w$, not $v$. This way, when a node $v$ runs \textsc{explore}, it only needs to pay at most $d(v)$ from its budget. 

    Now note that a node $v$ might run \textsc{explore} multiple times. In particular, $v$ runs \textsc{explore} either when we are processing $v$, or when a neighbor of $v$, say $u$ is being processed, and $d^{(u)}(v)<L/\pi(u)$. In the first case, we pay the cost of \textsc{explore} from $v$'s budget. In the second case, we pay the cost from $u$'s budget. Note that for a pivot node $v$, not only we have $d^{(v)}(v)<L/\pi(v)$, but also we have $d(v)<L/\pi(v)$ (see \cref{lem:pivot-set}) whp.

    So, in total, a pivot node $v$ only spends a budget of at most $d(v)\cdot \poly(\log n, 1/\eps)<L/\pi(v)\cdot \poly(\log n, 1/\eps)$: it spends $d^{(v)}(v)\le d(v)$ for scanning all its neighbors when it is being inserted, $d^{(v)}(v)\le d(v)$ when it runs \textsc{explore}, at most $ d(v)$ (possibly) when another pivot $w$ runs \textsc{explore} and removes $v$ from being a pivot, and $d^{(v)}(v)\poly(\log n, 1/\eps)\le d(v) \poly(\log n, 1/\eps)$ for running \textsc{break-cluster} (\cref{thm:break-cluster-running-time}).
    %\stodo{In place of $d^{(v)}(v)\le O(d(v) \poly(\log n, 1/\eps))$, there was $d^{(v)}(v)\le d(v)$. I guess you want the current one.}

    A non-pivot node $u$ spends at most $2L/\pi(u)$: once for scanning all its neighbors and once for paying for \textsc{explore} for its pivot $v$.

    Since $L=O(\log n)$, each node $u$ spends $\poly(\log n,1/\eps)+O(\log n/\pi(u))$. 
    Given that $\pi(u)$ is chosen uniformly at random from the range $[0, 1]$, in expectation, a node spends $\poly(\log n,1/\eps)\le \poly(\log T,1/\eps)$ running time.
\end{proof}
\section{Experiments-continued}\label{sec:extra-exp}
\subsection{Approximation guarantee results}
Below we include the approximation guarantee on the three remaining SNAP graphs.
\begin{figure}
    \centering
    \includegraphics[width=0.5\linewidth]{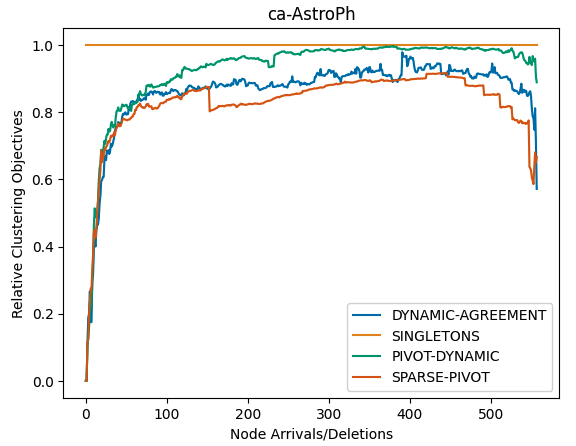}
    \caption{Correlation clustering cost for ca-astroph graph}
    \label{fig:ca-astroph}
\end{figure}
\begin{figure}
    \centering
    \includegraphics[width=0.5\linewidth]{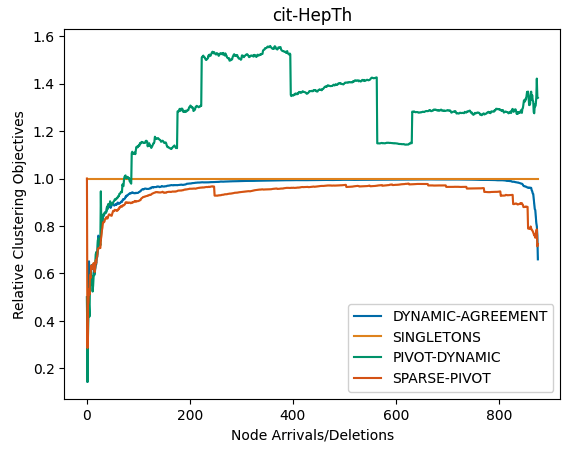}
    \caption{Correlation clustering cost for cit-hepth graph}
    \label{fig:cit-hepth}
\end{figure}
\begin{figure}
    \centering
    \includegraphics[width=0.5\linewidth]{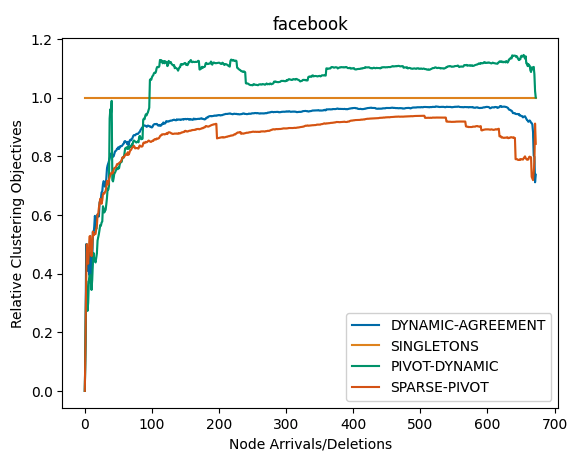}
    \caption{Correlation clustering cost for facebook graph}
    \label{fig:facebook}
\end{figure}
\subsection{Running Time}
Even though theoretically, our running time (and \textsc{dynamic-agreement} running time) is faster than \refClustering, this advantage only appears in very big graphs. This is because both algorithms have a lot of bookkeeping, which means that the constant behind the $O(\log n)$ running time is rather big. Nevertheless, we show in \cref{tab:drift-running-time} and \cref{tab:snap-running-time} that \ourPivot is faster than \textsc{dynamic-agreement}. 
\begin{table}
    \centering
    \begin{tabular}{c|c|c}
        Density & DA & SP \\ \hline
        $253.36$ & $36.75$ & $31.91$ \\
        $114.87$ & $43.08$ & $29.69$ \\
        $69.74$ & $50.77$ & $26.36$ \\
        $52.17$ & $49.27$ & $23.36$ \\
        $42.25$ & $41.23$ & $25.14$ \\
        \hline
        
    \end{tabular}
    \caption{Running time comparison on Drift dataset}
    \label{tab:drift-running-time}
\end{table}

\begin{table}
    \centering
    \begin{tabular}{c|c|c}
         Graph & DA & SP \\ \hline
         facecbook & $8.14$ & $3.31$ \\
        email-Enron & $9.69$ & $4.48$ \\
        cit-HepTh & $20.24$ & $11.53$  \\ 
        ca-AstroPh & $15.15$ & $3.84$\\ \hline
    \end{tabular}
    \caption{Running time comparison on SNAP datasets}
    \label{tab:snap-running-time}
\end{table}

%%%%%%%%%%%%%%%%%%%%%%%%%%%%%%%%%%%%%%%%%%%%%%%%%%%%%%%%%%%%%%%%%%%%%%%%%%%%%%%
%%%%%%%%%%%%%%%%%%%%%%%%%%%%%%%%%%%%%%%%%%%%%%%%%%%%%%%%%%%%%%%%%%%%%%%%%%%%%%%

\end{document}